\documentclass[a4paper,british,cleveref, autoref]{lipics-v2019}

\usepackage{mathtools}
\usepackage[all]{xy}
\usepackage{xspace}
\usepackage{graphics} 
\usepackage{float}
\usepackage{hyperref}

\newcommand{\red}[1]{{\color{red}#1}}
\makeatletter
\newcommand\thmlinebreak{\leavevmode\@beginparpenalty=10000}
\makeatother

\newcommand{\set}[1]{\left\{ #1 \right\}}

\newcommand{\structa}{{\mathfrak A}}

\newcommand{\mso}{{\sc mso}\xspace}
\newcommand{\fo}{{\sc fo}\xspace}

\newcommand{\Ii}{\mathcal I}

\newcommand{\mypic}[1]{
		\includegraphics[page=#1,scale=0.4]{interpretation-pics}
}

\newcommand{\mypicdomi}[1]{
	\begin{center}
		\includegraphics[page=#1,scale=0.365]{domination-pics}
	\end{center}
}

\newcommand{\Nat}{\mathbb N}

\graphicspath{{./graphics/}}

\title{String-to-String Interpretations with Polynomial-Size Output}

\titlerunning{String-to-String Interpretations}
 
\author{Miko\l aj Boja\' nczyk}{Institute of Informatics,\\ University of
Warsaw, Poland}{bojan@mimuw.edu.pl}{}{}

\author{Sandra Kiefer}{Department of Computer Science,\\ RWTH
Aachen University, Germany}{kiefer@cs.rwth-aachen.de}{}{}

\author{Nathan Lhote}{Institute of Informatics,\\ University of
Warsaw, Poland}{nlhote@mimuw.edu.pl}{}{}

\authorrunning{M. Boja\'nczyk and S. Kiefer and N. Lhote}

\Copyright{Miko{\l}aj Boja\'nczyk and Sandra Kiefer and Nathan Lhote}

\ccsdesc[500]{Theory of computation~Transducers}

\keywords{MSO, interpretations, pebble transducers, polyregular functions}

\acknowledgements{The authors would like to thank Benedikt Brütsch for helpful discussions on the topic.}%optional

\nolinenumbers %uncomment to disable line numbering

\hideLIPIcs  %uncomment to remove references to LIPIcs series (logo, DOI, ...), e.g. when preparing a pre-final version to be uploaded to arXiv or another public repository

\EventEditors{Christel Baier, Ioannis Chatzigiannakis, Paola Flocchini, and Stefano Leonardi}
\EventNoEds{4}
\EventLongTitle{46th International Colloquium on Automata, Languages, and Programming (ICALP 2019)}
\EventShortTitle{ICALP 2019}
\EventAcronym{ICALP}
\EventYear{2019}
\EventDate{July 9--12, 2019}
\EventLocation{Patras, Greece}
\EventLogo{eatcs}
\SeriesVolume{132}
\ArticleNo{102}
\begin{document}

\maketitle

\begin{abstract}
    String-to-string {\sc mso} interpretations are like Courcelle's \mso transductions, except that a single output position can be represented using a tuple of input positions instead of just a single input position. In particular, the output length is polynomial in the input length, as opposed to \mso transductions, which have output of linear length.  We show that string-to-string \mso interpretations are exactly the polyregular functions. The latter class has various characterizations, one of which is that it consists of the string-to-string functions recognized by pebble transducers. 

    Our main result implies the surprising fact that string-to-string \mso interpretations are closed under composition.
\end{abstract}

\section{Introduction}
A string-to-string function is called \emph{regular} if it is computed by a deterministic two-way automaton with output. There are many equivalent models for the same class of functions:  string-to-string \mso transductions~\cite{engelfriet2001mso}, streaming string transducers~\cite{alur2011streaming}, and various kinds of combinator-based formalisms~\cite{alur2014regular, DBLP:conf/lics/DaveGK18,DBLP:conf/lics/BojanczykDK18}.

A deterministic two-way automaton can visit each input position at most once in each state, otherwise it would loop forever. This means that  the length of the run -- and also the size of the output word -- is linear in the input string.  One way to go beyond linear-sized outputs was proposed by Milo, Suciu, and Vianu~\cite{milo2003typechecking}, following earlier work by Globerman and Harel~\cite{Globerman:1996he}: equip the automaton with $k$ pebbles which can be used to mark positions in the input word. To avoid making the model Turing-powerful, the pebbles are required to observe a so-called stack discipline: the pebbles are organised in a stack, and only the top-most pebble can be moved. In~\cite{DBLP:journals/corr/abs-1810-08760}, it is shown that pebble transducers are equivalent to  multiple other models: a higher-order functional programming language~\cite[Section 4]{DBLP:journals/corr/abs-1810-08760}, an imperative programming language with for-loops~\cite[Section 3]{DBLP:journals/corr/abs-1810-08760},  combinators~\cite[end of Section 4]{DBLP:journals/corr/abs-1810-08760}, and compositions of certain simple atomic functions~\cite[Section 1]{DBLP:journals/corr/abs-1810-08760}. Because of the multitude of models and their polynomial size outputs, the class of functions recognised by these models is called \emph{polyregular functions}. 

The list of models for polyregular functions described in~\cite{DBLP:journals/corr/abs-1810-08760} does not include any logical model. In this paper, we fix that omission. As mentioned above, for the regular functions, which have linear size output, the logical model consists in string-to-string \mso transductions. In an \mso transduction, each position of the output string is interpreted as a single position of the input string. A natural idea to capture polyregular functions is to  consider what we call \emph{string-to-string \mso interpretations}, where a position of the output string is represented by a $k$-tuple of positions in the input string. At first glance, this idea looks suspicious: if string-to-string \mso interpretations were equivalent to polyregular functions, then they would be closed under composition, because the class of polyregular functions is. However, composing two string-to-string \mso interpretations 
\begin{align*}
      \xymatrix{\Sigma^* \ar[r]^f & \Gamma^* \ar[r]^g & \Delta^*}
\end{align*}
raises the following issue. Suppose that positions of the intermediate word in $\Gamma^*$ are represented by $k$-tuples of positions in the input word from $\Sigma^*$. If an \mso formula defining $g$ quantifies over a set of positions in the intermediate word to define a property of the output word in $\Delta^*$, then this corresponds to quantifying over a set of $k$-tuples of positions in the input word. If we assume dimension  $k=1$, then the problem dissolves, and this is why \mso transductions have dimension $k=1$, whereas dimension $k>1$ is never used in the context of \mso (as opposed to first-order logic, where the standard notion of transformation, i.e.\ first-order interpretation, uses higher dimension). 

 As our main result, we show that the problems discussed above only invalidate the natural construction for composing \mso interpretations, which uses substitution of formulas. Still, and surprisingly, for structures that represent strings there exists a (less natural) construction. This follows from our main result which states that polyregular functions are exactly the string-to-string \mso interpretations. Indeed, corollaries of the main result are that (a) string-to-string \mso interpretations are closed under composition; and (b) for every regular string language, its inverse image under a string-to-string \mso interpretation is also regular. This is because (a) and (b) are true for polyregular functions. Proving (a) and (b) directly for string-to-string \mso interpretations seems hard; in fact an understandable (but wrong) first reaction to the claims (a) and (b) would be that they are false, for the reasons discussed in the previous paragraph. 

 It is easy to see that every polyregular function is a special case of a string-to-string \mso interpretation. One argument is that a $k$-pebble automaton can be simulated using  a string-to-string \mso interpretation, where configurations of the pebble automaton are represented using $k$-tuples of positions in the input word. The difficulty lies in proving the opposite direction and it comes from the stack discipline required in a pebble automaton. A $k$-tuple of positions used by an \mso interpretation can of course be viewed as a configuration of a pebble automaton, but there does not seem to be any reason why the resulting pebble automaton should observe stack discipline. It turns out -- and this is the main technical insight of this paper -- that any \mso formula which defines a linear ordering on $k$-tuples of positions  in strings must necessarily observe an implicit stack discipline, which makes it possible to translate a string-to-string \mso interpretation into a pebble automaton.

\subparagraph{Outline.} After describing string-to-string \mso interpretations in Section~\ref{sec:intepretations}, we revise polyregular functions via the formalism of for-programs in Section~\ref{sec:for-programs}. In Section~\ref{sec:equivalence}, we show that the models are equivalent.

\section{Interpretations}
\label{sec:intepretations}
In this section, we revise first-order and \mso interpretations, which are transformations of relational structures using formulas.

\subsection{Logic and interpretations}
\label{sec:interpretations-generally}
\subparagraph{Relational vocabularies and logic.}
A \emph{(relational) vocabulary} is a set of relation names, each one associated with a natural number called its \emph{arity}. For short, we refer to relational vocabularies simply as \emph{vocabularies}. A \emph{structure} over a vocabulary $\sigma$ consists of a set called the \emph{universe} and for each relation name of $\sigma$ a corresponding relation of the same arity over the universe. To define properties of relational structures, we use monadic second-order logic and its first-order fragment with the usual syntax and semantics \cite{Thomas97}. We use the convention that lower-case variables $x,y,z$ range over elements and upper-case variables $X,Y,Z$ range over sets of elements.

\subparagraph{Interpretations.} 
Intuitively speaking, an interpretation is a function from relational structures to relational structures where each element of the universe of the output structure is a tuple of elements of the input structure, and the relations of the output structure  are defined using formulas evaluated over the input structure. 
\begin{definition}[Interpretations over general structures] \label{def:interpretation} For $k \ge 1$, the syntax of a \emph{$k$-dimensional first-order  interpretation} consists of:
\begin{enumerate}
\item two vocabularies, called the  \emph{input vocabulary}  and the \emph{output vocabulary} 
\item \label{it:universe-formula} an \fo formula over the input vocabulary with $k$ free variables, called the \emph{universe formula}.
\item \label{it:relation-formula} for each $n$ and each $n$-ary relation name $R$ of the output vocabulary, an associated \fo formula $\varphi_R$ over the input vocabulary, with $k \cdot n$ free variables.
\end{enumerate}
\emph{\mso interpretations} are defined analogously, except that formulas of \mso are used, but the free variables still range over elements and not over sets.
\end{definition}

The semantics of an interpretation is a function from structures over the input vocabulary to structures over the output vocabulary, defined as follows.
\begin{itemize}
\item The universe of the output structure is the set of $k$-tuples of elements in the universe of the input structure which satisfy the universe formula from item~\ref{it:universe-formula} in Definition~\ref{def:interpretation}.
\item An $n$-ary relation name $R$ of the output vocabulary is interpreted as the set of $n$-tuples of $k$-tuples from the input structure, for which (a) each $k$-tuple is in the output universe, and (b) the entire $(n \cdot k)$-tuple satisfies the formula $\varphi_R$ in item~\ref{it:relation-formula} in Definition~\ref{def:interpretation}.
\end{itemize}

\subparagraph{Composition.}
First-order interpretations are closed under composition~\cite[p.~218]{Hodges:1993ki}. Let us recall the proof. Suppose that we want to compose interpretations
\begin{align*}
	\xymatrix@C=2cm{ \text{structures over $\sigma_1$} \ar[r]^{\Ii_1}& \text{ structures over $\sigma_2$ } \ar[r]^{\Ii_2}& \text{ structures over $\sigma_3$ } }
\end{align*}
of dimensions $k_1$ and $k_2$, respectively. The $(k_1 \cdot k_2)$-dimensional composition is obtained from $\Ii_2$ as follows: (a) quantification over elements of  $\Ii_2$ is replaced by a quantification over $k_1$-tuples of elements; and (b) relation names from $\sigma_2$ that appear in the input of $\Ii_2$ are replaced by the corresponding formulas from $\Ii_1$.  
This idea does not work for \mso in general, since set quantification in $\Ii_2$ would need to be replaced by quantification over sets of $k_1$-tuples. It does work when $k_1=1$. This essentially corresponds to Courcelle's transductions, for which closure under composition follows naturally~\cite[Theorem 7.14]{Courcelle:2012wq}. To recover closure under composition for $k_1 \geq 2$, one can use (not necessarily monadic) second-order logic, which by Fagin's Theorem \cite[Corollary 9.9]{libkin2013elements} corresponds to the polynomial hierarchy of computational complexity and is outside the scope of this paper.

\subsection{String-to-string interpretations} 
We are interested in interpretations that transform structures which represent strings. While there are two natural ways to model strings as relational structures, namely with an order relation or with a successor relation, only the order relation is useful in our context.  

 \begin{definition}[String-to-string interpretations]\label{def:string-to-string-mso-interpretations} For a  string $w \in \Sigma^*$, its \emph{ordered model} is defined to be the following relational structure, denoted by $\underline w$:
	\begin{itemize}
		\item the universe consists of the positions in the string, i.e., natural numbers;
		\item there is a binary relation for the natural order on positions;
		\item for each $a \in \Sigma$ there is a unary relation which is satisfied by every position with label $a$.
	\end{itemize}
A function $f \colon \Sigma^* \to \Gamma^*$ is called a \emph{first-order string-to-string interpretation} if the corresponding transformation on ordered models is a first-order interpretation for strings with length at least two\footnote{A typical operation we want to model is string duplication. When the input length is at least two, one can represent additional copies of the input string using a higher dimension. For input length  $n \le 1$, the output length will be  $n^k \le 1$ regardless of the dimension $k$. Another solution to this issue would be to have duplication built into the definition of interpretations. 
	}. 
 Likewise we define \emph{\mso string-to-string interpretations}.
	\end{definition}

	\begin{example}\label{ex:running-interpretation}
		Consider the function $f \colon \set{a,b}^*  \to \set{a,b}^*$ which maps a word to the concatenation of all of its reversed prefixes, as in the following example (with prefixes grouped for better readability):
		\begin{align*}
			abbb \quad \mapsto \quad \underbrace{a} \underbrace{ba} \underbrace {bba}  \underbrace{bbba}.
		\end{align*} 
		This transformation is the running example in~\cite{DBLP:journals/corr/abs-1810-08760}.
		We show that $f$ can be seen as a string-to-string first-order interpretation. The dimension is $2$, i.e.~positions in the output word represent pairs of positions in the input word. A pair $(x_1,x_2)$ of positions in the input word is used in the output word if it satisfies the universe formula $x_2 \le x_1$.
		The idea is that $x_1$ represents the length of the prefix, while $x_2$ is the position in that prefix. The label of a position $(x_1,x_2)$ is inherited from the second coordinate, as expressed by the formulas corresponding to labels on the output structure:
		\begin{align*}
			\varphi_a(x_1,x_2) = a(x_2) \quad \varphi_b(x_1,x_2) = b(x_2) 
		\end{align*}		
		The order on the positions of the output word is defined by the formula 
		\begin{align*}
			\varphi_{\le}(\underbrace{x_1,x_2}_{\substack{ \text{a position of}\\ \text{the output word}}},\underbrace{x'_1,x'_2}_{\substack{ \text{another position of}\\ \text{the output word}}}) \quad = \quad (x_1 < x'_1) \lor (x_1 = x'_1 \land x_2 \ge x'_2).
		\end{align*}
		Note that the above formula defines the lexicographic ordering on pairs of positions, with the first coordinate being used in increasing order, and the second coordinate being used in decreasing order. This, as it will turn out, is not a coincidence, since our main technical result says that it is impossible to define a linear order on tuples of positions without implicitly using some kind of lexicographic ordering. 
	\end{example}

\subparagraph{Successor instead of order.}
\newcommand{\succode}[1]{\underline{\underline{#1}}} When modelling a string as a relational structure, we use the order on positions. 
An alternative solution would be to use just the successor relation. The difference between the two solutions is that it is harder to define an order on $k$-tuples of positions than it is to define a successor relation. It turns out that the difference is crucial, and functions that output strings with successor can be ill-behaved. Note that whether or not the input string is equipped with an order or a successor relation makes no difference, since the order on the position of the input string can be recovered in \mso, which can compute the transitive closure of binary relations on positions.

Define the \emph{successor model} of a string in the same way as the ordered model from Definition~\ref{def:string-to-string-mso-interpretations}, except that a binary relation for successor is used instead of the order. 
  Define a \emph{successor-\mso string-to-string interpretation} to be a string-to-string function which is computed by an \mso interpretation, assuming that strings are represented by their successor models. Likewise, we define successor-first-order string-to-string interpretations. Successor-first-order string-to-string interpretations are closed under composition, because first-order interpretations are closed under composition. On the other hand, successor-\mso string-to-string interpretations are not closed under composition and lead to undecidability, as summarised in the following theorem. The proof can be found in Appendix \ref{app:successor}.

 \begin{theorem}\label{thm:successor}\thmlinebreak 
	\begin{enumerate}
		 \item \label{it:suc-mso-comp} The class of successor-\mso string-to-string interpretations is not closed under composition, and strictly contains the class of (order-)\mso string-to-string interpretations. 
		\item \label{it:suc-mso-undec} The following is undecidable: given a successor-first-order string-to-string interpretation $f$ and a regular language $L$ over the output alphabet, decide if $f^{-1}(L)$ is nonempty.
	\end{enumerate}
 \end{theorem}

\section{Polyregular functions}
\label{sec:for-programs}
Here we describe the class of polyregular functions. It has several equivalent characterisations, see~\cite[Theorem 4.4]{DBLP:journals/corr/abs-1810-08760}, one of which consists in the aforementioned pebble transducers. For the purposes of this paper, it will be most convenient to use a slightly more abstract characterisation in terms of for-programs, a machine model for string-to-string functions. We just explain the formalism on short examples, for a more detailed description see~\cite{DBLP:journals/corr/abs-1810-08760}.

\newlength{\twosubht}
\newsavebox{\twosubbox}

\begin{figure}[H]
\sbox\twosubbox{%
  \resizebox{\dimexpr\textwidth}{!}{%
    \mypic{1}%
    \mypic{5}
    \mypic{6}%
  }%
}
\setlength{\twosubht}{\ht\twosubbox}

\subcaptionbox{A for-program for the \\\phantom{\textbf{(a)}}function in Example~\ref{ex:running-interpretation}.\label{fig:for-program}}{%
  \hspace*{-0.12cm}\mypic{1}%
}\ 
\subcaptionbox{A for-program with a \\\phantom{\textbf{(b)}}Boolean variable {\tt P}.\label{fig:for-boolean}}{%
  \mypic{5}%
}  \ 
\subcaptionbox{A for-program which checks if \\\phantom{\textbf{(c)}} there is an {\tt a} between the \\\phantom{\textbf{(c)}} positions {\tt x1} and {\tt x2}.\label{fig:for-inputs-positions}}{%
  \mypic{6}%
}
\caption{Example for-programs.}\label{fig:example}
\end{figure}

Most of the syntactic constructions that can be used in a for-program are illustrated in Figure~\ref{fig:for-program}: (1) variables ranging over positions in the input word; (2) for-loops in which a variable iterates over all positions in the input word in increasing or decreasing order; (3) if-statements which depend on the order/labels of variables; (4) instructions which output letters. Position variables cannot be declared or written to, they are implicitly declared by for-loops and their only updates are the iterations performed by the for-loops.

The only feature of for-programs that is not used in Figure~\ref{fig:for-program} is (5) Boolean variables. Figure~\ref{fig:for-boolean} shows a program that outputs only those letters in the input word which have even distance to the last position. In the program, the Boolean variable {\tt P} is declared in the scope of a for-loop. On each iteration of the loop, the variable is reinitialised to false.

A for-program is called \emph{first-order definable} if Boolean variables can only be updated from {\tt false}, which is their initial value upon declaration, to {\tt true}. In other words, the only allowed update for Boolean variables is {\tt P := true}. For the first-order restriction, it is important that Boolean variables can be declared inside for-loops, and that they are reinitialised to {\tt false} at each iteration of the loop that they are declared in. The reason for the name ``first-order definable'' is that one can define in first-order logic the reachability relation on program states of the for-program, see~\cite[Lemma 5.3]{DBLP:journals/corr/abs-1810-08760}.

\begin{definition}
    A string-to-string function is called \emph{polyregular} if it is computed by a for-program. It is called \emph{first-order polyregular} if it is computed by a first-order definable for-program. 
\end{definition}

The class of polyregular functions has other characterisations, including the string-to-string pebble transducers introduced by Milo, Suciu and Vianu~\cite{milo2003typechecking}, as well as a higher-order functional programming language~\cite[Section 4]{DBLP:journals/corr/abs-1810-08760}. The main result of this paper, Theorem~\ref{thm:main} in the next section, adds a logical characterisation, namely string-to-string \mso interpretations.

\subparagraph{Evaluating first-order formulas.} 
The for-programs described above take as input strings and also output strings. One can also consider for-programs which input a string with distinguished positions and which output a Boolean value, as in Figure~\ref{fig:for-inputs-positions}. The distinguished positions are represented by free variables (here {\tt x1} and {\tt x2}) while the output value is taken from some distinguished Boolean variable, here {\tt P}. 

\begin{lemma}\label{lem:boolean-for-program} 
    Let $\varphi(x_1,\ldots,x_k)$ be an \fo formula over strings. There is a first-order for-program which computes the following.
    \begin{itemize}
        \item {\bf Input.} A word $w \in \Sigma^*$ and positions $x_1,\ldots,x_k$ in $w$;
        \item {\bf Output.} \textsf{Yes} or \textsf{No}, depending on whether $w$ satisfies $\varphi(x_1, \ldots, x_k)$.
    \end{itemize}
\end{lemma}
\begin{proof}
    The for-program implements the semantics of an \fo formula. For each quantifier, it loops over all possible values for the quantified position, and a Boolean variable is used to remember if some value has already been found which renders the formula true. 
\end{proof}

A similar result is true for \mso formulas, but the proof for that statement uses automata.

\section{Equivalence}
\label{sec:equivalence}
We show that the models defined in Sections~\ref{sec:intepretations} and~\ref{sec:for-programs} are equivalent.

\begin{theorem}\label{thm:main}\thmlinebreak 
    \begin{enumerate}
    \item \label{it:main-mso} String-to-string \mso interpretations are exactly the polyregular functions.
    \item  \label{it:main-fo} First-order string-to-string interpretations are exactly the first-order polyregular functions.    
    \end{enumerate}
\end{theorem}

Since the class of  polyregular functions is closed under composition\footnote{Closure under composition was proved for pebble transducers in~\cite[Theorem 11]{Engelfriet:2015eg} and for the class of for-programs in~\cite[Section 8.1]{DBLP:journals/corr/abs-1810-08760} as a step in proving equivalence with the other models of polyregular functions.}, we obtain:

\begin{corollary}
    String-to-string \mso interpretations are closed under composition. 
\end{corollary}

By using Theorem \ref{thm:main}, the proof of the corollary passes through for-programs. We are not aware of any direct proof that does not exploit the equivalence to polyregular functions.

The rest of this paper is dedicated to the proof of Theorem~\ref{thm:main}. We begin with a reduction of the first to the second item. This reduction illustrates a general phenomenon, namely that results about first-order polyregular functions often imply results about general polyregular functions, despite the latter class being larger. The reason behind this phenomenon is the following lemma, which says that for every polyregular function, all of the behaviour that is not first-order definable can be pushed into a simple preprocessing step. Define a \emph{rational function}, see~\cite[Section 13.2]{toolbox}, to be a string-to-string function which is recognised by a nondeterministic automaton, where every transition is labelled by a pair consisting of a letter from the input alphabet and a string over the output alphabet, and which is unambiguous in the sense that every input string admits exactly one accepting run. 

\begin{lemma}\label{lem:fo-reduction}\thmlinebreak 
    \begin{enumerate}
        \item \label{it:red-polyr} A function is polyregular if and only if it is a composition consisting of:
        \begin{enumerate}
            \item a (letter-to-letter) rational function; followed by 
            \item a first-order polyregular function.
        \end{enumerate}
        \item  \label{it:red-int}  A function is an \mso string-to-string interpretation if and only if it is a composition consisting of:
        \begin{enumerate}
            \item a (letter-to-letter) rational function; followed by 
            \item a first-order string-to-string interpretation.
        \end{enumerate}
    \end{enumerate}
\end{lemma}

The proof of Lemma \ref{lem:fo-reduction} is based on ideas from \cite{colcombet2007combinatorial,Kazana:2013jq,DBLP:journals/corr/abs-1810-08760} and uses factorisation forests.

\begin{proof}
    The right-to-left implications in items~\ref{it:red-polyr} and~\ref{it:red-int} are proved the same way: both polyregular functions and \mso string-to-string interpretations are closed under pre-composition with rational functions. For the class of polyregular functions, this holds because it is closed under composition and contains all rational functions~\cite[Theorem 1.6]{DBLP:journals/corr/abs-1810-08760}. For \mso string-to-string interpretations, one observes that rational functions are a special case of \mso string-to-string interpretations of dimension 1 (see~\cite[Figure 7]{filiot2016transducers}, where \mso interpretations of dimension 1 are the same as the so-called regular functions), and \mso interpretations are closed under pre-composition with such functions (see the remarks at the end of Section~\ref{sec:interpretations-generally}).

    To prove the left-to-right implications in items~\ref{it:red-polyr} and~\ref{it:red-int}, namely the decomposition into rational pre-processing and first-order post-processing, we use the following claim. 

    A \emph{letter-to-letter rational function} is a rational function where every transition in the underlying automaton is labelled with exactly one output letter, in which case the input and output strings have the same set of positions. 

    \begin{claim}
            Let $\varphi$ be an \mso formula which selects $k$-tuples of positions in strings over an alphabet $\Sigma$. There are a letter-to-letter rational function $f \colon \Sigma^* \to \Gamma^*$ and a first-order formula $\psi$ which selects $k$-tuples of positions in strings over the alphabet $\Gamma$ such that
            \begin{align*}
                w \models \varphi(\bar x) \quad \text{iff} \quad f(w) \models \psi(\bar x) \qquad \text{for every $w \in \Sigma^*$ and $k$-tuple of positions $\bar x$.}
            \end{align*}
    \end{claim}
    The claim is the special case of~\cite[Theorem 2]{colcombet2007combinatorial} for finite strings instead of infinite trees, and its proof uses factorisation forests (see ~\cite{simon1990factorization}). Another proof of the above claim is in~\cite[Theorem 3.2]{Kazana:2013jq}. Using the claim, we immediately get the left-to-right implications in item~\ref{it:red-int}.
    For item~\ref{it:red-polyr}, we also use Lemma~\ref{lem:fo-reduction} to obtain a first-order for-program realizing the function. The main idea is that if the reachability relation is first-order definable, then one can define a first-order query which accepts consecutive produced tuples.
\end{proof}

With the lemma, we show that item~\ref{it:main-fo} in Theorem~\ref{thm:main} implies item~\ref{it:main-mso}, i.e.~if first-order string-to-string interpretations are exactly the first-order polyregular functions, then \mso interpretations are exactly the polyregular functions:
\begin{align*}
    \text{polyregular} &=& \text{\small by Item~\ref{it:red-polyr} of Lemma~\ref{lem:fo-reduction}}\\
     \text{(first-order polyregular)} \circ \text{rational} &=& \text{by Item~\ref{it:main-fo} of Theorem~\ref{thm:main}} \\
     \text{(first-order interpretations)} \circ \text{rational} &=& \text{by Item~\ref{it:red-int} of Lemma~\ref{lem:fo-reduction}} \\
     \text{\mso interpretations}
\end{align*}

It remains to prove item~\ref{it:main-fo} in Theorem~\ref{thm:main}, i.e.~that first-order string-to-string interpretations are exactly the first-order polyregular functions. The right-to-left inclusion follows immediately from~\cite[Lemma 5.3]{DBLP:journals/corr/abs-1810-08760}, which says that a formula in first-order logic can define the reachability relation on program states in first-order for-programs. We are left with the left-to-right-inclusion:
    \begin{align}\label{eq:inclusion-main}
        \text{first-order string-to-string interpretations} \  \subseteq \  \text{first-order definable for-programs}
    \end{align}

The rest of the paper is devoted to showing the above inclusion. 
When simulating a first-order interpretation by a for-program, we will mainly be concerned with the universe of the output string (which is a set of $k$-tuples of positions in the input string) and its ordering. The labelling of the $k$-tuples can then be recovered using the for-program from Lemma~\ref{lem:boolean-for-program}. The main result is that every first-order definable linear ordering on tuples of positions can be implemented by a for-program. 
To be able to speak about this result, we introduce some notation for devices that produce lists of tuples of positions.

\subparagraph{Enumerators.}  
 Let $k \in \Nat$. A \emph{$k$-enumerator} over an alphabet $\Sigma$ is a function of the following form:
\begin{itemize}
    \item {\bf Input.} A string $w \in \Sigma^*$;
    \item {\bf Output.} A list of $k$-tuples of positions in $w$, which is nonrepeating\footnote{Every tuple appears at most once, but positions can appear in multiple tuples. We need this for the existence of the formulas stated in the following definitions.}.
\end{itemize}
We compare the following two ways of implementing $k$-enumerators:
\begin{enumerate}
        \item A $k$-enumerator is called \emph{definable} if there are two \fo formulas: one with $k$ variables, which says when a tuple is part of the output list, and one with $2k$ variables, which defines a total order on the tuples selected by the first formula. 
        \item A $k$-enumerator is called \emph{programmable} if its output can be computed by a first-order for-program which instead of outputting letters uses instructions of the form {\tt output (x1,...,xk)} where {\tt x1}$,\ldots, ${\tt xk} are position variables. 
    \end{enumerate}
    For definable $k$-enumerators, the order on tuples in the output list is given explicitly by the formula $\varphi$, while in programmable ones, the order is implicit from the order in which the output instructions are executed during the computation. 
    
    \begin{example} We present an enumerator based on Example~\ref{ex:running-interpretation}.
        Consider the $2$-enumerator which outputs all pairs of positions $(x_1,x_2)$ with $x_2 \le x_1$, listed in lexicographic order, where $x_1$ is ordered in increasing order and $x_2$ is ordered in decreasing order. Here is an example:
    \begin{align*}
        abbb \quad \mapsto \quad  (1,1), (2,2), (2,1), (3,3), (3,2), (3,1), (4,4), (4,3), (4,2), (4,1)
    \end{align*}
    This enumerator is definable, as witnessed by the formula $\varphi_{\le}$ in Example~\ref{ex:running-interpretation}.
     The formula $\varphi_{\le}$ is quantifier-free, but in general, quantifiers are allowed. Here is a for-program which computes the same function:
     \begin{center}
     \mypic{4}
     \end{center}
    \end{example}
    The following lemma is the main technical result of this paper. 
\begin{lemma}\label{lem:main-lemma} Every  
     definable $k$-enumerator is also programmable.
\end{lemma}

Our proof of Lemma~\ref{lem:main-lemma} uses two fundamental ingredients. The first is by now standard: this is Simon's factorisation forest theorem~\cite{simon1990factorization}, which roughly says that every string can be cut into pieces that are similar to each other. The second ingredient is new: the \hyperref[lem:domination]{Domination Lemma}, presented in Section~\ref{sec:domination-overview}, roughly says that if a string is cut into pieces that are similar to each other, then any first-order definable linear order on tuples of positions must observe an implicit stack discipline. These two results are combined in Section~\ref{sec:use-of-domination} to prove Lemma~\ref{lem:main-lemma}. 
Before we proceed with the proof of Lemma~\ref{lem:main-lemma}, we use it to complete the proof of Theorem \ref{thm:main}.

\begin{proof}[Proof of Theorem \ref{thm:main}, second part] The only part of Theorem~\ref{thm:main} that has not been proved yet is that every first-order string-to-string interpretation is  polyregular. Suppose that $f$ is a $k$-dimensional first-order string-to-string interpretation. Consider the $k$-enumerator which inputs a string $w$ and outputs the list of $k$-tuples of positions in $w$ that are used to represent output positions of $f(w)$, in the appropriate order. Apply Lemma~\ref{lem:main-lemma} to obtain a first-order for-program $g$ which computes the same list. To compute the original function $f$, we use a for-program which behaves as $g$, except that instead of outputting a $k$-tuple of positions like $g$, it uses the program described in Lemma~\ref{lem:boolean-for-program} as a subroutine to check what is the output letter that should be produced for this tuple, and outputs that letter. 
\end{proof}

\subsection{The Domination Lemma}
\label{sec:domination-overview} 
In this section we present the Domination Lemma, which says that if $\prec$ is a first-order definable linear order on $k$-tuples of positions in a string, then there is an implicit stack discipline in the following sense. For every type (see below) $t$ of tuples of positions there is a coordinate $d \in \set{1,\ldots,k}$ such that for the subset of $k$-tuples of positions consisting in all of type $t$, the order $\prec$ is uniquely determined by the order of the $d$-th coordinates in the string.

We begin by explaining the notions of types. 
For $r \in \set{0,1,\ldots}$,  the \emph{rank $r$ type} of a structure $\structa$ with $k$ distinguished positions $\bar x \coloneqq (x_1, \dots, x_k)$ is defined to be the set of first-order formulas of quantifier rank at most $r$ and $k$ free variables that are true in $\structa, \bar x$. The number $k$ is the \emph{arity} of the type. For arity 0, we talk about the rank $r$ \emph{type of the structure $\structa$}. If the structure $\structa$ is implicit from the context, then we talk about the rank $r$ type of the tuple $\bar x$. For every finite vocabulary, there are finitely many types of given arity and rank. We write $\equiv_r$ for the equivalence relation on structures with distinguished elements of having the same rank $r$ type.
     For a binary relation $R$, its \emph{inverse} is the set $\{(v,u) \mid (u,v) \in R\}$. For $p \in \set{1,-1}$, define $R^p$ to be either $R$ or its inverse, depending on the value of $p$.

    \begin{lemma}[Domination Lemma]\label{lem:domination}
         For all $k,m,r \in \set{1,2,\ldots}$, there is an $\omega \in \set{1,2,\ldots}$ with the following property. Let $n \in \{1,2, \ldots\}$, let
$w_1,\dotsc, w_n$ be strings over some alphabet $\Sigma$ 
        and let $\structa$  be the ordered structure of the concatenation $w_1 \cdots w_n$ extended with the \emph{block order} defined by
        \begin{align*}
            x \sqsubset y \qquad \text{if} \qquad \text{$x$ is a position in $w_i$ and $y$ is a position in $w_j$ for some $i<j$.}
        \end{align*}
         Let $\prec$ be a linear order on $k$-tuples in $\structa$ defined by a first-order formula of quantifier rank $r$, and let $t$ be a $k$-ary rank $\omega$ type over the vocabulary of $\structa$. If 
        \begin{align*}
            w_i \equiv_\omega w_{i+1} \qquad \text{holds for all $i \in \set{1,\ldots,n-1}$ with at most $m$ exceptions,}
        \end{align*}
        then there  is a $d \in \set{1,\ldots,k}$, called the \emph{dominating coordinate}, and a $p \in \set{-1,1}$, called the \emph{polarity}, such that 
        \begin{align*}
             x_d \sqsubset^p y_d \quad \text{implies} \quad (x_1,\ldots,x_k) \prec (y_1,\ldots,y_k) \qquad \text{for all 
              $\underbrace{x_1,\ldots,x_k}_{\text{of type $t$}}, \,\underbrace{y_1,\ldots,y_k}_{\text{of type $t$}}$ in $\structa$}.
        \end{align*}
    \end{lemma}

The \hyperref[lem:domination]{Domination Lemma} is the technical heart of this paper. The full proof is presented in Appendix \ref{sec:domination-lemma}. To explain more intuitively some of the ideas that we use, we treat a special case in detail. In the \hyperref[lem:domination]{Domination Lemma}, the structure $\structa$ consists of blocks organised in a linear way.  A very simple linear order -- although infinite -- is the natural one on the rational numbers; one reason for its simplicity is that quantifiers can be eliminated (see \cite[Section 5.6.2]{DBLP:series/txtcs/GradelKLMSVVW07}). Because of this, it is quite easy to prove a version of the \hyperref[lem:domination]{Domination Lemma} for the rational numbers and still its proof bears some similarity to the proof of the general case. 

\newcommand{\growrat}{\mathbb Q^{(k)}}
\begin{lemma}[Rational Domination Lemma]  Let $\prec$ be a linear ordering on $k$-tuples of rational numbers defined by a quantifier-free (equivalently, first-order) formula using only the usual ordering $<$ on rational numbers. Then there is a coordinate $d \in \set{1,\ldots,k}$ and a polarity $p \in \set{-1,1}$ such that
    \begin{align*}
        x_d <^p y_d \quad \text{implies} \quad  (x_1,\ldots,x_k) \prec (y_1,\ldots,y_k)
   \end{align*}
for all tuples of rational numbers satisfying $x_1 <\cdots <x_k$ and  $y_1<\cdots<y_k$.
 \end{lemma} 
\begin{proof} We first prove the statement for $k=1$ and $k = 2$ and then we deduce the general case.  
    \begin{enumerate}
        \item When $k=1$, then the formula defining $\prec$ must be either $x < y$ or $x > y$. 
        \item For $k=2$, we do a case analysis. Note that whether $\bar x \prec \bar y$ or $\bar y \prec \bar x$ holds depends only on the order relationship of the positions in $\bar x$ and $\bar y$ in the rational numbers and not on the precise values in $\bar x$ and $\bar y$.

        The following picture shows the two possible relationships for two pairs $\bar x$ and $\bar y$ when they are ``consecutive'' and the two possible relationships when they are ``nested'':

        {\hspace{-1.5cm}{\mypic{7}}}

        Suppose we are given a pair $\bar x$ and without loss of generality, assume the ``consecutive growing'' case for a second pair $\bar y$. We only show the proof for the case that there is a pair $\bar y'$ such that $\bar x$ and $\bar y$ are ``nested growing'' (``nested decreasing'' works analogously). We prove that $d=1$ is dominating for $\prec$ with polarity $p = 1$. Consider all three remaining configurations of pairs $\bar x$ and $\bar y$ with $x_1 < y_1$. In all cases, $\bar x \prec \bar y$ is proved by finding an intermediate pair (drawn in yellow), whose order with respect to $\bar x$ and $\bar y$ follows from the assumptions ``consecutive/nested growing'' (in the pictures below, we assume that lower lines represent bigger tuples in the ordering $\prec$):

        {\centering{\mypic{9}}}
         
        \item Consider the case  $k > 2$. Fix a ``growing'' tuple of $k$ rational numbers, i.e.\ a tuple $\bar z$ such that for $1 \leq i < j \leq k$ it holds that $z_1 \leq z_i < z_j \leq z_k$. Define $\prec^{\bar z}_{ij}$  to be the  restriction of $\prec$  to tuples that agree with $\bar z$ on coordinates from $\set{1,\ldots,k} \setminus \set{i,j}$. 
        Using the reasoning from the previous item, the ordering $\prec^{\bar z}_{ij}$ must admit some  dominating coordinate $d \in \set{i,j}$ and one of the cases ``growing'' or ``decreasing''. This must hold for every choice of $\bar z$ and $i,j$. Furthermore, the dominating coordinate $d$ depends only on $i$ and $j$ and not on $\bar z$, likewise for the choice of ``growing'' or ``decreasing''. Let us write $i \to j$ if $j$ dominates, otherwise we write $j \to i$. The reasoning in the following picture shows  that $\to$ is transitive, i.e.~$i \to j$ and $j \to m$ implies $i \to m$:

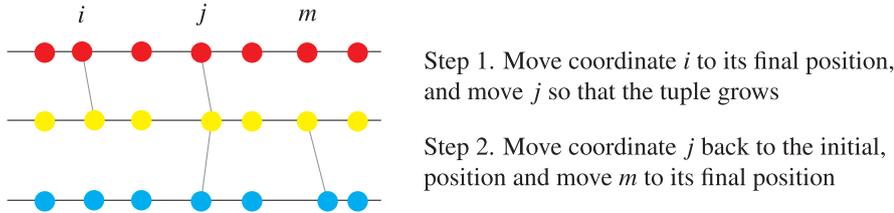
        \mypic{11}
        
        Therefore, $\to$ is in fact a total order on $\set{1,\ldots,k}$. Let $d$ be the maximum with respect to this order. The following picture explains why $d$ is the dominating coordinate $d$ from the statement of Lemma \ref{lem:domination}.

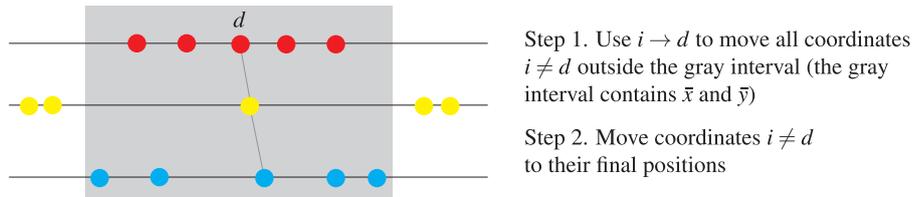
         \begin{center}
            \includegraphics[page=12,scale=0.36]{interpretation-pics}
        \end{center}

    Suppose without loss of generality that we are in the ``growing'' case for each pair of coordinates. Then we can first move all coordinates apart from $d$ to positions smaller than $\min\{x_1,y_1\}$ or bigger than $\max\{x_k,y_k\}$ and then use the dominations $i \to d$ to move them, one by one, to their final positions (always increasing the $d$-th coordinate slightly to a value in the open interval $(x_d,y_d)$).
    \end{enumerate}
\end{proof}

\subsection{Proof of Lemma~\ref{lem:main-lemma}}
\label{sec:use-of-domination}
We now return to Lemma~\ref{lem:main-lemma}, i.e., we prove that every definable $k$-enumerator is also programmable. In the proof, we use the following version of the Factorisation Forest Theorem. We use the term \emph{interval} for a connected set of positions in a string. 

\begin{theorem}[Factorisation Forest Theorem, aperiodic variant] \label{thm:simon}
    Let $h \colon \Sigma^+ \to S$ be a semigroup homomorphism, where $S$ is finite and aperiodic. Then there exists a function which assigns to each string in $\Sigma^+$ a partition of the  positions into intervals (so-called \emph{blocks}) such that: 
    \begin{enumerate}
        \item All blocks are nonempty, and for each string in $\Sigma^+$ of length at least 2, there are at least two blocks.
        \item \label{it:same-blocks} If a string has at least three blocks, then all of the blocks have the same value under $h$.
        \item \label{it:small-depth} There exists $M \in \Nat$ such that all strings have height at most $M$, where the height of a string is defined as follows: letters have height 1, for other strings  the height is the maximum of the heights of its blocks + 1. 
        \item  \label{it:fo-definable} There is a first-order formula $\varphi$ such that for every string $w$, the positions satisfying $\varphi(x)$ are exactly the first positions of the blocks of $w$.
    \end{enumerate}
\end{theorem}

Apart from the \hyperref[thm:simon]{Factorisation Forest Theorem} and the \hyperref[lem:domination]{Domination Lemma}, our proof uses the following straightforward result on combining outputs of two for-programs. 
As a convention, if $\psi$ is a first-order formula with $k$ free variables and $f$ is a $k$-enumerator, then $f|\psi$ denotes the $k$-enumerator where the output list of $f$ is filtered so that it contains only tuples satisfying $\psi$.

\begin{lemma}[Merging Lemma]\label{lem:merge}
   Let $f$ be a definable $k$-enumerator. Let $\Phi$ be a finite set of \fo formulas $\psi$, each one with $k$ free variables, such that every $k$-tuple of positions satisfies at least one formula from $\Phi$. Then $f$ is programmable if and only if every $f | \psi$ is programmable. 
\end{lemma}

\begin{proof}
    For the left-to-right implication, we observe that the filtering $f| \varphi$ can be implemented by a for-program thanks to Lemma~\ref{lem:boolean-for-program}. We are left with the right-to-left implication. It suffices to examine the case $|\Phi| = 2$. The general case follows by a straightforward induction. 

    Suppose that $f|\varphi_1$ and $f|\varphi_2$ are implemented by programs $f_1,f_2$.
    First check whether the first tuple in the output satisfies $\varphi_1$ or $\varphi_2$ using the result from Lemma~\ref{lem:boolean-for-program} and an if-statement, and then run a different program for each of the two outcomes.
    By symmetry, it suffices to consider inputs where the first tuple in the output of $f|(\varphi_1 \vee \varphi_2)$ satisfies $\varphi_1$. Take the code of $f_1$, and
         after each instruction which outputs a tuple of positions $\bar x$, run a copy of the code for $f_2$, with its output restricted to tuples $\bar y$  which satisfy:
        \begin{itemize}        
            \item $\bar y$ is after $\bar x$ according to $f$; and
                    \item  there are no other tuples from the output of $f_1$ between $\bar x$ and $\bar y$.
        \end{itemize}
        The first item can be checked by a for-program using the assumption that $f$ is definable and Lemma~\ref{lem:boolean-for-program}, while the second item can be checked by running a nested copy of $f_1$.
\end{proof}

We are now ready to prove Lemma~\ref{lem:main-lemma}.
Let $f$ be a definable $k$-enumerator. We need to describe a for-program which outputs the same list of tuples as $f$. Let $r$ be the maximal quantifier rank of the first-order formulas used in the definition of $f$. Apply the \hyperref[lem:domination]{Domination Lemma} to $k$, $m \coloneqq 5k$, and $r$, yielding a constant $\omega$. 
  Define $h$ to be the function which maps a string $w \in \Sigma^+$ to the rank $\omega$ type of the corresponding ordered model of $w$. Compositionality of first-order logic (see~\cite[Section 3.4]{libkin2013elements}) on strings says that the image of $h$, the set of rank $\omega$ types of strings, is a finite aperiodic semigroup and $h$ is a semigroup homomorphism. Apply the Factorisation Forest Theorem to $h$, yielding a function which partitions each string into blocks and an upper bound $M$ on heights of strings.
By abuse of notation, we lift notions about strings to intervals inside strings: the height of an interval $X$ in a string $w$ is defined to be the height (in the sense of item~\ref{it:small-depth} in Theorem \ref{thm:simon}) of the infix of $w$ induced by $X$. Likewise, we define the blocks of $X$ as the blocks of the infix induced by $X$, viewed as intervals contained in $X$.

To show that $f$ is also programmable, we use an induction over heights in factorisation forests. More precisely, we prove that for every $i \in \Nat$ there is a for-program which computes the following:
    \begin{itemize}
        \item {\bf Input.} A string $w \in \Sigma^+$ with distinguished nonempty intervals $X_1,\ldots,X_k$ that are pairwise equal or disjoint, and such that the sum of their heights (in the sense of Theorem~\ref{thm:simon}) is at most $i$. Each interval is represented by its first and its last position. 
        \item {\bf Output.} The list $f(w)$ restricted to tuples in $X_1 \times \cdots \times X_k$. 
    \end{itemize}

By item~\ref{it:small-depth} in Theorem \ref{thm:simon}, the for-program with parameter $i \coloneqq kM$ will work for every choice of pairwise equal or disjoint intervals, in particular when all of the intervals are the entire string. The induction base $i=k$ (where every interval has the height $1$) is straightforward: each interval is a singleton, and the for-program simply checks if the unique tuple in $X_1 \times \cdots \times X_k$ belongs to the output of $f$ by using the subroutines from Lemma~\ref{lem:boolean-for-program}. The rest of the proof is devoted to the induction step, more specifically, to producing the correct order of the tuples: whether a tuple belongs to the output or not can again be checked using the subroutines from Lemma~\ref{lem:boolean-for-program}.

 Let $X_1,\ldots,X_k$ be intervals in an input string $w$  that are pairwise disjoint or equal.  Define $\mathcal X$ to be the coarsest partition of the positions in the input string into intervals that satisfies $X_1,\ldots,X_k \in \mathcal X$.   This partition uses at most $2k+1$ intervals. Consider a factorisation 
\begin{align*}
    w = w_1 \cdots w_n
\end{align*}
where each $w_j$ is a block of one of the elements of $\mathcal X$. 
Define $\structa$ as in the \hyperref[lem:domination]{Domination Lemma}, i.e.~as the ordered structure of $w$ extended with an extra order $\sqsubset$ that describes the partition into factors $w_1,\ldots,w_n$. By item~\ref{it:fo-definable} of the \hyperref[thm:simon]{Factorisation Forest Theorem}, the  order $\sqsubset$ can be defined by a first-order formula which uses the input string and the endpoints of the intervals $X_1,\ldots,X_k$. It follows that for every $k$-ary rank $\omega$ type $t$ over the vocabulary of $\structa$, there is a corresponding first-order formula which selects the $k$-tuples of positions in $w$ that have type $t$ in $\structa$. Since there are finitely many choices of $t$, it follows from the \hyperref[lem:merge]{Merging Lemma} that it is enough to show that for every $t$, there is a for-program which outputs the tuples of type $t$. 

Let $t$ be a $k$-ary rank $\omega$ type over the vocabulary of $\structa$. We show a for-program which outputs all tuples in 
\begin{align*}
    T \coloneqq \set{\bar x  \in X_1 \times \cdots \times X_k : \text{$\bar x$ has type $t$ and is in the output of $f(w)$}}
\end{align*}
according to their order given by $f(w)$, call this order $\prec$.  

If an interval from $\mathcal X$ has more than two blocks, then, by item~\ref{it:same-blocks} of the \hyperref[thm:simon]{Factorisation Forest Theorem}, all of these blocks have the same image under $h$, i.e., the same rank $\omega$ type. Since there are at most $2k+1$ intervals, it follows that with at most $2(2k+1)-1 = 4k+1 < 5k$ exceptions, consecutive strings $w_j$ and $w_{j+1}$ have the same rank $\omega$ type. Hence, for the order $\prec$ defined by $f(w)$, the \hyperref[lem:domination]{Domination Lemma} yields $d \in \set{1,\ldots,k}$ and $p \in \set{-1,1}$ such that
\begin{align*}
     x_d \sqsubset^p y_d \quad \text{implies} \quad  (x_1,\ldots,x_k) \prec (y_1,\ldots,y_k) \qquad \text{for all  $\underbrace{x_1,\ldots,x_k}_{\text{of type $t$}},\, \underbrace{y_1,\ldots,y_k}_{\text{of type $t$}}$ in $\structa$}.
\end{align*}
This means that the tuples in $T$ are $\prec$-ordered as $T_1 \prec^p T_2 \prec^p \cdots \prec^p T_{s}$, 
where $s$ is the number of blocks in $X_d$ and $T_j$ consists of the tuples from $T$ where the coordinate $x_d$ is in the $j$-th block of $X_d$. Our for-program can simply loop over all the blocks of $X_d$ -- in increasing or decreasing order depending on the choice of $p$ -- because the endpoints of each block can be identified in first-order logic due to item~\ref{it:fo-definable} of the \hyperref[thm:simon]{Factorisation Forest Theorem}. In each iteration of the loop, the for-program outputs the tuples in the corresponding $T_j$ using the following claim, thus completing the proof of the lemma. 

\begin{claim}\label{claim:block-induction} There is a for-program which inputs the $i$-th block of $X_d$, given by its endpoints, and outputs the tuples from $T_j$ ordered according to $\prec$.
\end{claim}
\begin{proof}[Proof of the claim]  The general idea is to replace $X_d$ with its $j$-th block (call this block $X$) and use the induction assumption. However, if there is an $i \neq d$ such that $X_j = X_d$, then replacing $X_d$ with $X$ would violate the assumption that the intervals are pairwise disjoint or equal (since $X \subsetneq X_j$). To overcome this issue, we use the following simple case disjunction. For each of the $3^k$ possible values of 
    \begin{align*}
        v \in \set{\text{positions before $X$, $X$, positions after $X$}}^k
    \end{align*}
    construct a for-program that outputs all tuples from $Y_1  \times \cdots \times Y_k$, where $Y_j$ is the intersection of $X_j$ with the $j$-th entry of $v$. Since each $Y_j$ is a union of blocks of $X_j$, it is empty or its height is at most the height of $X_j$. Furthermore, if $Y_d$ is nonempty, then it is $X$, which is a block of $X_d$, and therefore its height is strictly smaller than the height of $X_d$.  It follows that the induction assumption can be applied to produce all tuples in $Y_1 \times \cdots \times Y_k$, for any given choice of $v$. These choices can be combined using the \hyperref[lem:merge]{Merging Lemma}.
    \end{proof}

\bibliography{main}
\newpage
\appendix
%!TEX root = main.tex
\section{Successor instead of order}
\label{app:successor}

In this appendix, we prove Theorem~\ref{thm:successor}, which says that:
\begin{enumerate}
	\item \label{it:suc-mso-comp0} The class of successor-\mso string-to-string interpretations is not closed under composition, and strictly contains the class of (order-)\mso string-to-string interpretations.
	\item \label{it:suc-mso-undec0} The following is undecidable: given a successor-first-order string-to-string interpretation $f$ and a regular language $L$ over the output alphabet, is $f^{-1}(L)$ is empty?
\end{enumerate}

\begin{proof}[Proof of Theorem~\ref{thm:successor}]\ 
\begin{enumerate}
	\item We first show item \ref{it:suc-mso-comp0}. Fix an input alphabet $\Sigma$, and consider the function $f \colon \Sigma^* \to (\Sigma \times \Sigma)^*$ which inputs a string, and outputs all pairs of positions (with the corresponding pairs of labels) in the order depicted by the following picture:

	\begin{center}
	\mypic{13}
	\end{center}

	 It is not hard to see that the function $f$ is a succesor-\mso string-to-string interpretation (in fact even first-order logic would be enough if the input string was equipped with a labelling indicating the parity of positions). Suppose that the alphabet $\Sigma$ contains two endmarkers $\vdash, \dashv$, and consider an input word of the form $\vdash a_1 \cdots a_n \dashv$ where $a_1,\ldots,a_n$ are letters that are not  endmarkers and the length $n$ is even. In this case, the output contains the letter $(\vdash,\dashv)$ exactly once, it contains  the letter $(\dashv,\vdash)$ also exactly once, and the word between these two letters is exactly:
	 \begin{align*}
		 (a_1,a_n), \ldots, (a_n,a_1).
	 \end{align*}
	If and only if the word $a_1 \cdots a_n$ is a palindrome, then the above word contains only letters from the diagonal $\set{(a,a) : a \in \Sigma}$. Summing up, there is a regular (and therefore also \mso-definable) language $L \subseteq \Sigma^*$ such that 
	\begin{align}\label{eq:palindrome}
		f(w) \in L \qquad \text{if and only if} \qquad w = \ \vdash\!\! v\!\!\dashv \text{ for some palindrome $v$ without $\vdash,\dashv$ }.
	\end{align}
	Define $\chi_L$ to be the \emph{characteristic function} of $L$, i.e., the function from $\Sigma^*$ to $\set{0,1}$ which outputs $1$ or $0$ depending on whether the input belongs to $L$ or not. We can view the characteristic function as a string-to-string function, where the output is in $\set{0,1}^*$ and which happens to only produce outputs with one letter. 
	The following claim is not hard to see.
	
	\begin{claim}\label{claim:char-fun}
		A language $L \subseteq \Sigma^*$ is regular if and only if its characteristic function is a successor-\mso string-to-string function.
	\end{claim}
	
	From the claim, it follows that the characteristic function of the language $L$ in~\eqref{eq:palindrome} is in the successor-\mso class. If the class were closed under composition, then also $\chi_L\circ f$, the characteristic function of the palindrome language in~\eqref{eq:palindrome}, would be in successor-\mso, and thus by Claim~\ref{claim:char-fun} the palindrome language would be regular, a contradiction.
	\item We now show item~\ref{it:suc-mso-undec0} of Theorem~\ref{thm:successor}, i.e., that for a successor-first-order string-to-string interpretation $f$ and a regular language $L$ over the output alphabet, the emptiness of $f^{-1}(L)$ is undecidable. The proof is a standard reduction from the halting problem for Turing machines.

	 Let $M$ be a Turing machine. Consider the string-to-string function $f$ defined as in the previous item, except that the order on positions is as follows:

	\begin{center}
	\mypic{15}
	\end{center}

	%\mypic{14}

	The key observation is that the output $f(a_1 \cdots a_n)$ contains, for every odd $i \in \set{1,\ldots,n}$, an infix of the form
	\begin{align*}
		(a_i,a_{1}),(a_{i+1},a_{2}),\ldots,(a_{n},a_{n-i+1}).
	\end{align*}
	In the picture, the blue colouring indicates this infix for $i=3$.

	The above observation shows that the output of $f$ can be used to compare infixes of $w$ with other infixes; this can be used to check if an input word represents an accepting computation of the fixed Turing machine.

	The input will be required to be of the following shape: $|c_1|c_2|\ldots|c_n|$, where the $c_i$s are words that represent the consecutive configurations of an accepting computation of the Turing machine.

	We mainly need to enforce two additional properties to obtain the reduction: first, that all the $c_i$s have the same size and second, that each $c_{i+1}$ is the successor configuration of $c_i$ (and also that $c_1$ is initial and $c_n$ is final, which are simple regular properties).
	To enforce these two properties we only need to check properties of the infix $(a_i,a_{1}),(a_{i+1},a_{2}),\ldots,(a_{n},a_{n-i+1})$ where $i$ is the position of the second $|$ separator symbol of the input word. We can easily enforce that this position is odd by asking that all configurations are of even length.
\end{enumerate}
\end{proof}

The proof of item~\ref{it:suc-mso-undec0} could be improved so that the function $f$ is a successor-first-order string-to-string interpretation, which shows that emptiness of $f^{-1}(L)$ is undecidable already when $f$ is successor-first-order and $L$ is regular. This shows that the class of successor-first-order string-to-string interpretations is not contained in the class of (ordered) first-order string-to-string interpretations considered in this paper, since by our main theorem, the latter class is contained in the class of polyregular functions, and emptiness of $f^{-1}(L)$ is decidable if $L$ is regular and $f$ is polyregular~\cite[Theorem 1.7]{DBLP:journals/corr/abs-1810-08760}.
%!TEX root = main.tex
\section{Proof of the Factorisation Forest Theorem}
We provide a proof for the aperiodic variant of the Factorisation Forest Theorem (Theorem \ref{thm:simon}) here. Consider a surjective homomorphism
\begin{align*}
    h \colon \Sigma^+ \to S.
\end{align*}
We can assume without loss of generality that $\Sigma$ is a subset of $S$. 
The proof is by induction on (a) the size of $S$; (b) the size of $\Sigma$. The two parameters are ordered lexicographically, with (a) being more important. \

When $\Sigma$ has one element, then the blocks of a string $w \in \Sigma^+$ are simply its letters; this covers the induction base.  The partition of a 
string into letters is clearly first-order definable. 

For the induction step, suppose that $\Sigma$ has more than one element. Take some  $s \in \Sigma$ and consider the functions
\begin{align}\label{eq:side-mult}
    \begin{array}{ccc} S &\rightarrow& S \\ t & \mapsto & ts  \end{array} \qquad \text{and}\qquad \begin{array}{ccc} S &\rightarrow& S \\ t & \mapsto & st  \end{array}
\end{align}

If one of these functions is surjective, then it is a permutation, and therefore it has to be the identity by aperiodicity of the semigroup. If both functions are surjective, then  $s$ must be the identity element of the semigroup (which might not exist in some semigroups. Since $\Sigma$ has at least two elements, and there is at most one identity, there must be an $s \in \Sigma$ such that one of the functions in~\eqref{eq:side-mult} is not surjective. Without loss of generality, assume that $t \mapsto ts$ is not surjective, and therefore $T \coloneqq Ss$ is a proper subset  of the semigroup $S$. 

Consider the following two semigroup homomorphisms: the first one is the product operation
\begin{align*}
    h_T \colon T^+ \to T
\end{align*}
in the semigroup $T$, and the second one is 
\begin{align*}
    h_{\neq s} \colon (\Sigma - \set s)^+ \to S
\end{align*}
obtained by restricting $h$ to the  smaller alphabet. Both homomorphisms are smaller in our induction order: $h_T$ uses a smaller semigroup, and $h_{\neq s}$ has a smaller alphabet. Therefore, the induction assumption can be applied to obtain both partitions into blocks and bounds $M_T$ and $M_{\neq s}$ on the heights of the corresponding strings.

For a string $w \in \Sigma^+$, we define its partition into blocks with respect to the homomorphism $h$ as follows by case analysis.

\begin{enumerate}
    \item \label{it:blocks} Suppose that $w$ ends with $s$ and does not begin with $s$.  Decompose $w$ as follows:
    \begin{align*}
        w = w_1 s^{k_1} \cdots w_n s^{k_n}  \qquad w_1,\ldots,w_n \in (\Sigma - \set s)^+ \qquad k_1,\ldots,k_n \in \set{1,2,\ldots}.
    \end{align*}
    \begin{enumerate}
        \item \label{it:length-one-t} If $n=1$, then the blocks are $w_1$ and $s^{k_1}$. The former word has height at most $M_{\neq s}$ by induction assumption and the latter word has height at most 2 because it uses only the letter $s$. It follows that $w$ has height at most $M_{\neq s}+2$. 
        \item Otherwise $n >1$. For $i \in \set{1,\ldots,n}$ define $t_i$ to be $h(w_i s^{k_i})$. Note that $t_i \in Ss=T$. Consider the partition into blocks of the word $t_1 \cdots t_n$ with respect to the homomorphism $h_T$.  The blocks of $w$ are the same as the blocks of $t_1 \cdots t_n$, except that in each block, the letter $t_i$ is replaced with the corresponding infix $w_i s^{k_i}$. Since the height of $t_1\cdots t_n$ is at most $M_T$ from the induction assumption, and each $w_i s^{k_i}$ has height at most $M_{\neq s}+2$ from item~\eqref{it:length-one-t}, we obtain a height of at most $M_T+M_{\neq s}+2$ for the word $w$.
    \end{enumerate}
    \item We are left with the case when $w$ either begins with $s$ or does not end with $s$. In these cases, we simply decompose the word by shaving off the beginning and the end to reduce the decomposition to case~\eqref{it:blocks}.
    \begin{enumerate}
        \item If $w=usv$ such that $u$ and $v$ do not begin with $s$, and $v\in (\Sigma - \set s)^+$ then we split $w$ into two blocks, $us$ and $v$. According to case~\eqref{it:blocks}, $us$ has height at most $M_T+M_{\neq s}+2$ and by induction assumption, $v$ has height at most $M_{\neq s}$, thus overall, $w$ has height at most $M_T+M_{\neq s}+3$.
        \item Finally, let $w=s^kusv$ with $k\in \set{1,2,\ldots}$, $u,v$ not beginning with $s$, and $v\in (\Sigma - \set s)^+$. In that case we split $w$ into two parts again: $s^k$ and $usv$, $s^k$ has height at most 2, and from the previous case, we have a final height of at most  $M_T+M_{\neq s}+4$.
    \end{enumerate}

\end{enumerate}
It is not hard to see the partition into blocks described above is first-order definable. This completes the proof of the \hyperref[thm:simon]{Factorisation Forest Theorem}. 
%!TEX root = main.tex

\newcommand{\structb}{\mathfrak B}
 
\section{Proof of the Domination Lemma}
\label{sec:domination-lemma}
This section is devoted to proving the \hyperref[lem:domination]{Domination Lemma}.  The statement of the Domination Lemma in Section~\ref{sec:equivalence} was chosen so that it would be most easily applied to strings and their infixes. We begin by stating a more abstract version of the lemma, called the \emph{Product Domination Lemma}, which is adapted to allow for a modular proof and implies the \hyperref[lem:domination]{Domination Lemma} in the shape in which we use it. Before stating the Product Domination Lemma, we introduce notation for the three kinds of product operations that are relevant to us.
\begin{enumerate}
    \item Elements of the \emph{direct product $\prod_{i = 1}^k \structa_i \coloneqq \structa_1 \times \cdots \times \structa_k$} of structures $\structa_1$, \dots, $\structa_k$ are tuples $(a_1,\ldots,a_k)$ with $a_i \in \structa_i$ for every $i \in \set{1, \ldots,k}$. For every relation $R$ in some $\structa_i$, there is a corresponding relation of the same arity in the direct product, which says whether or not $R$ holds after projecting to the $i$-th coordinate. 
    \item The \emph{$k$-th power of a structure $\structa$} is similar to the $k$-fold direct product of $\structa$, except that different coordinates can be compared, i.e., for every two tuples $(a_1, \ldots, a_k), (a'_1, \ldots, a'_k)$ in the $k$-th power of $\structa$ and all $i, j \in \{1, \dots, k\}$, we can compare $a_i$ and $a'_j$.
    One way of modelling such comparisons is to say that the $k$-th power is obtained from the $k$-fold direct product by adding for all $i,j \in \set{1,\ldots,k}$ a function which swaps coordinates $i$ and $j$. 
    \item The \emph{ordered product $\structa_1 \cdots \structa_k$} is obtained by taking the disjoint union of the structures $\structa_1,\ldots,\structa_k$ and adding an extra binary predicate $\sqsubset$, called the \emph{block order}, such that $x \sqsubset y$ holds if $x$ comes from an $\structa_i$ and $y$ comes from an $\structa_j$ with $i<j$. 
\end{enumerate}

The Product Domination Lemma uses all three kinds of products: it considers a direct product of powers of ordered products.

Recall that we write $\equiv_{r+k}$ for the equivalence relation on structures with distinguished elements of having the same rank $r+k$ type. 
    
    \begin{lemma}[Product Domination Lemma]\label{lem:technical-domination}
        Let $k,r \in \set{1,2,\ldots}$. Then there exists an $\omega_* \in \set{1,2,\ldots}$ such that the following holds. Let $I\subset \set{0,1,\ldots}$ be an initial segment of the natural numbers and let
        \begin{align*}
            \structa \coloneqq \prod_{i \in I} \structa_i^{k_i} \qquad \text{such that $k_i \le k$ for all $i$},
        \end{align*}
        where each $\structa_i$ is an ordered product
        \begin{align*}
      \structa_i = \structa_{i,1} \cdots \structa_{i,n_i}  \qquad \text{with }\structa_{i,1} \equiv_{r+k} \cdots \equiv_{r+k} \structa_{i,n_i}.
        \end{align*}
     Let $\prec$ be a linear order on $\structa$ defined by a first-order formula of rank $r$, and let $t$ be a unary rank $\omega_*$ type over $\structa$. Then there  exist $d \in I$,  $e \in \set{1,\ldots,k_d}$,  and   $p \in \set{-1,1}$ such that  
        \begin{align*}
            x[d][e] \sqsubset^p  y[d][e] \quad \text{implies} \quad   x \prec  y \qquad \text{for all $x,y \in \structa$ of type $t$.} 
        \end{align*}
    \end{lemma}

\subparagraph{Proof overview.} We begin by showing, in Section~\ref{sec:final-proof-of-domination}, that the \hyperref[lem:technical-domination]{Product Domination Lemma} implies the \hyperref[lem:domination]{Domination Lemma} in its original form from Section~\ref{sec:equivalence}. The rest of Section~\ref{sec:domination-lemma} is then devoted to proving the \hyperref[lem:technical-domination]{Product Domination Lemma}. This is done in four steps. 
In Section~\ref{sec:gap-reduction}, we show that if we can prove the \hyperref[lem:technical-domination]{Product Domination Lemma} for some nonzero polarity other than $\set{-1,1}$, then we can reduce the polarity down to $\set{-1,1}$ at the cost of increasing the threshold $\omega$. Next, we prove the \hyperref[lem:technical-domination]{Product Domination Lemma} in four steps, which deal with special cases of increasing generality, as described below.
\begin{itemize}
    \item In Section~\ref{sec:direct-linear} we prove domination for direct products of linear orders, i.e.~structures  
    \begin{align*}
        \structa = \prod_{i \in I} (\set{1,\ldots,n_i},<).
    \end{align*}
    This can be viewed as the special case of the \hyperref[lem:technical-domination]{Product Domination Lemma} when all $k_i$ are $1$, and furthermore all structures $\structa_{i,j}$ (called  \emph{blocks} in the proof) have size one.
    \item In Section~\ref{sec:linear-domination} we prove domination for powers of linear orders, i.e.~structures 
    \begin{align*}
        \structa = (\set{1,\ldots,n},<)^k.
    \end{align*}
    This can be viewed as the special case of the \hyperref[lem:technical-domination]{Product Domination Lemma} when $I$ has size one, and all blocks have size one.
    \item In Section~\ref{sec:products-of-powers-of-linear-orders} we prove the joint generalisation of the results from the two previous sections, i.e.~we consider direct products of powers of linear orders:
    \begin{align*}
        \structa = \prod_{i \in I} (\set{1,\ldots,n_i},<)^{k_i}.
    \end{align*}
    This can be viewed as the special case of the \hyperref[lem:technical-domination]{Product Domination Lemma} when all blocks have size one.
    \item In Section~\ref{sec:pre-final-proof-of-domination} we complete the proof of the \hyperref[lem:technical-domination]{Product Domination Lemma}. 
\end{itemize}

\subparagraph{Compositionality.} Before continuing with the proof, we state two compositionality properties of first-order logic with respect to products that will be heavily used in the proofs.

    \begin{theorem}[\cite{Mostowski:1952ew}] \label{thm:mostowski}
    The following holds for all $m,n,r \in \set{1, 2, \ldots}$.
        \ 
        \begin{enumerate}
            \item Consider structures $\structa_1,\ldots,\structa_n,\structb_1,\ldots,\structb_m$ over the same vocabulary. The  rank $r$ type of the ordered product
            \begin{align*}
                \structa_1 \cdots \structa_n \structb_1 \cdots \structb_m
            \end{align*}
            is determined by the rank $r$ types of the two ordered products   
            \begin{align*}
                \structa_1 \cdots \structa_n \qquad \text{and} \qquad \structb_1 \cdots \structb_m.
            \end{align*}
            \item Let $\structa \coloneqq \prod_{i \in I}\structa_i$ be a direct product of structures $\structa_i$. 
            For every $m$-ary rank $r$ type $t$ in $\structa$ and every $i \in I$ there is an $m$-ary rank $r$ type $t[i]$ in the structure $\structa_i$ such that for all $x_1,\ldots,x_m \in \structa$
            \begin{align*}
                (x_1,\ldots,x_m) \text{ has rank $r$ type $t$ in $\structa$} \qquad \text{iff} \qquad &\text{for every $j$, $(x_1[j],\ldots,x_m[j])$ has}
                \\&\text{type $t[j]$ in $\structa_j$}.
            \end{align*}
        \end{enumerate}
        \end{theorem}
    
    \subparagraph{Continuous functions.} Let $i \in \set{0,1,\ldots}$ and let $\structa$ and $\structb$ be  relational structures over possibly different vocabularies.
    Let $f$ be a function from (the universe of) $\structa$ to (the universe of) $\structb$ and for all $k \in \set{1,2,\ldots}$, denote by $f_k$ the function mapping $k$-tuples of $\structa$ to $k$-tuples of $\structb$ by component-wise application of $f$. Then $f$ is called \emph{$i$-continuous} if for every $k \in \set{1,2,\ldots}$ and every subset of $\structb^k$ defined by a formula in first-order logic with quantifier rank $r$, its inverse image under $f_k$ can be defined in first-order logic via a formula with quantifier rank $r + i$. A function is called \emph{continuous} if it is $0$-continuous.

    An alternative, equivalent characterisation, which we also use in this paper, is that a function is continuous if and only if it is type-preserving, i.e., it maps tuples of the same type to tuples of the same type.

%!TEX root = main.tex

\subsection{Proof of the Domination Lemma}
\label{sec:final-proof-of-domination}
We begin by using the \hyperref[lem:technical-domination]{Product Domination Lemma} to obtain the \hyperref[lem:domination]{Domination Lemma} in its statement from Section~\ref{sec:equivalence}.
Let $k, m, r \in \{1,2, \dots\}$. Apply the \hyperref[lem:technical-domination]{Product Domination Lemma} to $k$ and $r$ yielding some threshold value $\omega_*$. Define 
\begin{align}\label{eq:rkm}
    \omega \coloneqq 2\omega_* + r + k + m.
\end{align}
We prove that $\omega$ satisfies the requirements of the \hyperref[lem:domination]{Domination Lemma}. 
  Let $w_1,\ldots,w_n$ and $\structa$  
be as in the assumptions of the \hyperref[lem:domination]{Domination Lemma}. This means that $\structa$ is the ordered product of the (ordered structures associated with) the strings $w_1,\ldots,w_n$ extended with the block order $\sqsubset$. Furthermore, the strings satisfy $w_i \equiv_\omega w_{i+1}$ with at most $m$ exceptions. (The order on the blocks is $\sqsubset$, while the orders corresponding to the ordered structures are $<$). Let $\prec$ be a linear order on $\structa^k$ defined by a first-order formula of quantifier rank $r$ and let $t$ be a $k$-ary rank $\omega$ type over the vocabulary of $\structa$. We intend to find a dominating coordinate and a polarity that satisfy 
        \begin{align*}
             x_d \sqsubset^p y_d \quad \text{implies} \quad (x_1,\ldots,x_k) \prec (y_1,\ldots,y_k) \qquad \text{for all 
              $\underbrace{x_1,\ldots,x_k}_{\text{of type $t$}}, \,\underbrace{y_1,\ldots,y_k}_{\text{of type $t$}}$ in $\structa$}.
        \end{align*}     

Define $\sim$ to be the coarsest equivalence relation on $\set{1,\ldots,n}$ such that $i \sim i+1$ holds whenever $w_i \equiv_{\omega_*} w_{i+1}$. Equivalence classes of $\sim$ are intervals. Let $\mathcal I$ be the set consisting of these equivalence classes. Since $\omega \ge \omega_*$, we know from the assumptions of the \hyperref[lem:domination]{Domination Lemma} that $\mathcal I$ has at most $m$ elements. For an equivalence class $I \in \mathcal I$, define $\structb_I \subseteq \structa$ to be the substructure obtained by restricting $\structa$ to elements that come from  $w_i$ with  $i \in I$. We can view $\structb_I$ as an ordered product which only uses the words $w_i$ with $i \in I$.
By definition, in $\structb_I$, all blocks (i.e., all $w_i$) have the same rank $\omega_*$ type. 
For every $I \in \Ii$, there is a first-order formula which selects the elements from $\structb_I$ inside the structure $\structa$: the formula counts the number of blocks $w_i$ to the left which satisfy $w_i \not \equiv_{\omega_*} w_{i+1}$, and therefore it has quantifier rank at most $\omega_* + m$. For $\bar x \in \structa^k$ of rank $\omega$ type $t$ and $I \in \Ii$, define
\begin{align*}
    C_I \coloneqq \set{ i \in \set{1,\ldots,k} : \text{$\bar x[i]$ is in $\structb_I$}}.
\end{align*}
This set does not depend on $\bar x$ once $t$ has been fixed, because, as we have argued above, one can express the containment in $\structb_I$ using a first-order formula with  quantifier rank at most $\omega$. Define
\begin{align*}
\iota \colon   \overbrace{\prod_{I \in \Ii} \structb_I^{C_I}}^\structb \to \structa^k
\end{align*}
to be the injection that is defined in the following way (where an element in the universe of $\structb_I$ is seen as an element in the universe of $\structa$):
\begin{align*}
\iota(x)[i] \coloneqq x[I][i],\qquad \text{for } i\in \set{1,\ldots,k}, \text{ and } i\in C_I.
\end{align*}

 The image of this injection contains all tuples in $\structa^k$ that have $k$-ary rank $\omega$ type $t$. Since the injection sends tuple of the same type to tuples of the same type, it is continuous.
 \begin{claim} All elements in the inverse image of $t$ under $\iota$ have the same rank $\omega_*$ type. 
 \end{claim}
 \begin{proof} Note that the continuity of $\iota$ is not useful for this result, because it only tells us that the inverse image of $t$ is a union of rank $\omega$ types. 
     The image $\iota(\structb) \subseteq \structa^k$ can be defined by a first-order formula of quantifier rank at most $\omega_* + r+k+m$. Therefore, if elements of $\structb$ have different rank $\omega_*$ types, then their images under $\iota$ have different rank $\omega$ types. This proves the claim.
 \end{proof}
By the above claim, all elements in the inverse image under $\iota$ of type $t$ have the same rank $\omega_*$ type over $\structb$, call it $t_\structb$.
 Define  $\prec_\structb$ to be the linear order on $\structb$ which is the inverse image of $\prec$ under $\iota$, i.e.
\begin{align*}
    x \prec_\structb y \qquad \text{iff} \qquad \iota(x) \prec \iota(y).
\end{align*}
Since $\iota$ is continuous, it follows that $\prec_\structb$ is defined 
using a first-order formula of quantifier rank $r$. 
By the \hyperref[lem:technical-domination]{Product Domination Lemma}, there exist $d \in \Ii$,  $e \in C_d$and $p \in \set{-1,1}$ such that
\begin{align*}
    x[d][e] \sqsubset^p y[d][e] \quad \text{implies}\quad  x \prec_\structb y  \qquad \text{for all $x,y$ of type $t_\structb$ in $\structb$.}
\end{align*}
By pulling this result forward across the injection $\iota$, we get the corresponding conclusion for $\bar x$ and $\bar y$ in $\structa^k$ of type $t$.

This finishes the proof of the \hyperref[lem:domination]{Domination Lemma}, assuming the \hyperref[lem:technical-domination]{Product Domination Lemma} holds. 
The rest of this section is devoted to proving the \hyperref[lem:technical-domination]{Product Domination Lemma}. 

%!TEX root = main.tex
\subsection{Polarity reduction}
\label{sec:gap-reduction}

In the \hyperref[lem:technical-domination]{Product Domination Lemma}, we use powers $\sqsubset^p$ for polarities $p \in \set{-1,1}$. This notation also makes sense for other nonzero integers $p$, for example $x \sqsubset^{-3} y$ means that $y \sqsubset z_1 \sqsubset z_2 \sqsubset x$  holds for some $z_1$, $z_2$. (We extend this notation to other binary relations as well.) 
 It would be easier to prove the \hyperref[lem:technical-domination]{Product Domination Lemma} for polarities $p$ with larger absolute values, since for $s \in \set{-1,1}$ and for $p' \in \set{1,2,\ldots}$, the implication 
\begin{align*}
     x[d] \sqsubset^{s \cdot p'}  y[d] \quad \text{implies} \quad   x \prec  y \qquad \text{for all $ x,  y \in \structa$ of type $t$}
\end{align*}
has a stronger assumption than $x[d] \sqsubset^s  y[d]$ and is therefore weaker than
\begin{align*}
     x[d] \sqsubset^s  y[d] \quad \text{implies} \quad   x \prec  y \qquad \text{for all $ x,  y \in \structa$ of type $t$}.
\end{align*}
The following lemma shows that such weaker versions are indeed enough.

\begin{lemma}[Polarity Reduction Lemma]\label{lem:polarity-reduction}
    Let $\structa$ be a relational structure, and let $R$ and $\prec$  be binary relations on $\structa$ that are defined by first-order formulas of quantifier rank at most $r \in \set{1,2,\ldots}$.  If $\prec$  is transitive and antisymmetric, then for every $p \in \set{1,2,\ldots}$
    \begin{eqnarray*}
        R(x,y) \quad \text{implies} \quad (x \prec y \lor y \prec x) &\qquad& \text{for all }x,y \in \structa \\
        & \land \\
        R^p(x,y) \quad \text{implies} \quad x \prec y &\qquad& \text{for all }x,y \in \structa \\
        & \Downarrow \\
        R(x,y) \quad \text{implies} \quad x \prec y &\qquad& \text{for all }x,y \in \structa  \text{  with }x \equiv_{r+p} y.\\
    \end{eqnarray*}
\end{lemma}

\begin{proof}
Let $R$ and $\prec$ be as in the assumptions and let $p \in \set{1,2,\ldots,}$. Suppose the two conditions in the stated implication $\Downarrow$ hold. Let $x,y \in \structa$ be such that $R(x,y)$ and $x \equiv_{r+p} y$ hold. We need to show $x \prec y$. 
Let $t$ be the binary rank $r$ type that describes the pair $(x,y)$. Since $R$ is defined using quantifier rank at most $r$ and contains $(x,y)$, it follows that all pairs of type $t$ are contained in $R$. Because $\prec$ is defined by a formula of quantifier rank $r$, our assumptions imply that the set of pairs of type $t$ is contained in either $\prec$ or $\succ$. To prove the lemma, we need to show that it is contained in $\prec$. Define a \emph{chain} to be a sequence of elements 
\begin{align*}
    x_1,\ldots,x_i \in \structa \qquad \text{where }(x_1,x_2),\ldots,(x_{i-1},x_i) \text{ have type $t$},
\end{align*}
i.e.~a walk in the directed graph on the universe of $\structa$ where $t$ is the edge relation. Note that every chain is either growing or decreasing with respect to $\prec$. We need to rule out the ``decreasing'' case. 
The property ``there is a chain of length $i$ that begins in $x$'' can be defined by a first-order formula of quantifier rank $r+i$. It follows inductively from $t(x,y)$ and $x \equiv_{r+p} y$ that there is a chain which begins in $x$ and has length at least $p$. Indeed, if the maximal length of a chain beginning in $y$ was some value $p' < p$, there would be a chain of length $p' + 1$ beginning in $x$ since $(x,y)$ has type $t$. This would violate that $x \equiv_{r+p} y$.
\end{proof}

As discussed at the beginning of this section, a corollary of the \hyperref[lem:polarity-reduction]{Polarity Reduction Lemma} is that it is enough to prove a weaker version of the \hyperref[lem:technical-domination]{Product Domination Lemma}, where the polarity $p$ from the conclusion is in $\set{-\omega_*,\omega_*}$ instead of $\set{-1,1}$. To see why, suppose that we have proved the version of the \hyperref[lem:technical-domination]{Product Domination Lemma} with polarity $p' \in \set{-\omega,\omega}$, and we want to prove the version with polarity $p \in \set{-1,1}$. 
Let $t$ be a unary rank $w_*$ type. By the weaker version of the \hyperref[lem:technical-domination]{Product Domination Lemma}, there is some $p \in \set{-1,1}$ such that 
\begin{align*}
     x[d] \sqsubset^{p \cdot \omega_*}  y[d] \quad \text{implies} \quad   x \prec  y \qquad \text{for all $ x, y$ in $\structa$ of type $t$}.
\end{align*}
Apply the \hyperref[lem:polarity-reduction]{Polarity Reduction Lemma} for $p \coloneqq \omega_*$, the structure being $\structa$, and the relation $R$ defined by 
\begin{align*}
    R( x,  y)  \quad \text{iff} \qquad  x[d] \sqsubset^p  y[d].
\end{align*}
We obtain the conclusion of the \hyperref[lem:technical-domination]{Product Domination Lemma} in its original form. 

Thanks to the above reasoning, in the remaining sections it suffices to show variants of the \hyperref[lem:technical-domination]{Product Domination Lemma} where the polarity $p$ in the conclusion is some nonzero number with a fixed upper bound, not necessarily 1, on its absolute value.

%!TEX root = main.tex

\subsection{Direct products of linear orders}
\label{sec:direct-linear}
In this section, we show the special case of the \hyperref[lem:technical-domination]{Product Domination Lemma} for direct products of linear orders.  A corollary is going to be that every first-order definable ordering $\prec$ in a product of linear orders coincides with a lexicographic product of the underlying orders, at least when restricted to elements of the direct product that have the same type. The corollary is stated later in this section, but we begin with the underlying result about dominating coordinates.

\begin{lemma}\label{lem:lex-product}
Let $r \in \set{1,2,\ldots}$ and let  $\prec$ be a linear ordering on a direct product
    \begin{align*}
        \structa = \prod_{i \in I} (\set{1,\ldots,n_i},<) \qquad \text{where $n_i > 6 \cdot 2^r$ for every $i \in I$}
    \end{align*}
    such that $\prec$ is defined by a first-order formula of quantifier rank $r$. Then for every unary rank $r$ type $t$ in $\structa$, there is a dominating coordinate $d \in I$ and $p \in \set{-2 \cdot 2^{r}, 2 \cdot 2^{r}}$ such that
    \begin{align*}
        x[d] <^p y[d]  \quad \text{implies} \quad x \prec y  \qquad \text{for all $x,y \in \structa$ of type $t$.}
    \end{align*} 
\end{lemma}
Note that in the above lemma, the polarity $p$ can have a value not contained in $\set{-1,1}$. As explained in Section~\ref{sec:gap-reduction}, at the cost of increasing the quantifier rank of the type $t$, the polarity can be reduced to values in $\set{-1,1}$.

To prove Lemma \ref{lem:lex-product}, we use the following observation, which expresses that first-order logic formulas with quantifier rank $r$ can only measure distances up to $2^r$ in a linear order. Its proof is the same as for~\cite[Theorem 3.6]{libkin2013elements}. 
\begin{lemma}[Threshold Lemma]\label{lem:threshold}
Let $r \in \set{1,2,\ldots}$ and consider two $k$-tuples $ x$ and $ y$ in a linearly ordered set
\begin{align*}
    (\set{1,\ldots,n},<)
\end{align*}
     Then $ x$ and $ y$ have the same rank $r$ type if and only if they have the same quantifier-free type in the structure extended with relations $<^i$ for $i \in \set{1,\ldots,2^r}$and unary relations $min$ and $max$.
\end{lemma}

\begin{proof}[Proof of Lemma~\ref{lem:lex-product}] The proof proceeds by induction on the size of the set $I$, i.e.~on the dimension of the product.  Let $t$ be a unary rank $r$ type over the vocabulary of $\structa$. Consider its projections $t[i]$ for $i \in I$ as in Theorem \ref{thm:mostowski}. 
    By the \hyperref[lem:threshold]{Threshold Lemma}, if an $|I|$-tuple $x \in \structa$ has type $t$, then $t[i]$ determines the distance of $x[i]$ from the first and last positions in $\set{1,\ldots,n_i}$, measured up to threshold $2^r$. If the distance from either the first or last position is $<2^r$, then the value of $x[i]$ is fixed by the type $t$.
    For such types, we can eliminate one coordinate, and obtain the result by using the induction assumption.
    We are left with the case when  $t$ expresses that all coordinates are at least $2^r$ positions away from both the first and last positions. 

    Choose a tuple $x \in \structa$ such that for every $i \in I$, coordinate $x[i]$ is at least $3 \cdot 2^r$ positions away from the first and last positions in $\set{1,\ldots,n_i}$, which can be achieved by the assumption that  $n_i \ge 6 \cdot 2^r$ for all $i \in I$.
     We say that $\delta \in  \mathbb Z^I$ is  \emph{small} if for every $i \in I$, the absolute value of $\delta[i]$ is between $2^r$ and $2 \cdot 2^{r}$. 
    By the choice of $x$, we know that if $\delta$ is small, then $x + \delta$ has the same type as $x$.  Define the \emph{sign vector} of $\delta$ to be the subset of $I$ which contains the coordinates on which $\delta$ is positive, and define 
    \begin{align*}
        \mathcal I \coloneqq \set{\text{sign vector of $\delta$} : \text{$\delta$ is small and $x \prec x  +\delta$}} \subseteq 2^I.
    \end{align*}
   It is not hard to see that the family $\Ii$ is closed under set union, and the same is true for its complement. 

    \begin{claim}
         Consider a partition of the powerset $2^I$ into two families of sets, both of which are closed under union. Then there is a $d \in I$ such that one of the families is $\set{J \subseteq I : d \in J}$.
    \end{claim}
    \begin{proof}
        Since both families are closed under union, it follows from De Morgan's law that both families are closed under intersection. One of the families does not contain the empty set, call this family $\Ii$. Since $\Ii$ is closed under intersection, it follows that the intersection $\cap \Ii$ of all sets in $\Ii$ is nonempty. The interesction $\cap \Ii$ cannot have more than one element, because otherwise it could be decomposed as a union of two sets outside $\Ii$, and therefore it would be outside $\Ii$. Hence, the intersection of all sets in $\Ii$ is a singleton $\set d$ for some $d \in I$. This means that all sets in $\Ii$ contain $d$. It follows that for every $i \neq d$, the singleton $\set{i}$ must belong to the complement of $\Ii$. Since this complement is closed under taking unions, every set that does not contain $d$ belongs to the complement of $\Ii$. 
    \end{proof}

    An application of the claim yields a coordinate $d \in I$ such that either $\Ii$ or its complement consists in exactly the sets that contain $d$. By symmetry, we may assume the first case. By unfolding the definition of $\Ii$, it follows that incrementing $x[d]$ by at least $2^r$ and at most $2 \cdot 2^r$ and modifying all other coordinates by any number with absolute value between $2^r$ and $2 \cdot 2^{r}$ yields a bigger tuple. Performing this procedure twice allows us to modify the coordinates other than $d$ by any value in $\set{-2^r,\ldots,2^r}$, and therefore the result follows using the \hyperref[lem:threshold]{Threshold Lemma}.
\end{proof}

We end this section with an interesting consequence of Lemma~\ref{lem:lex-product}. Define \emph{a lexicographic ordering} on 
\begin{align*}
    \structa = \prod_{i \in I} (\set{1,\ldots,n_i},<) 
\end{align*}
to be an ordering that is the lexicographic product, under some ordering of $I$, of orderings in the coordinates that are either $<$ or $>$. If $I$ has $n$ elements, then there are $n! \cdot 2^n$ lexicographic orderings, since the coordinates can be ordered in $n!$ ways and for each ordering one can use $<$ or $>$ for each of the $n$ coordinates. 
By iteratively applying Lemma \ref{lem:lex-product} and then using the \hyperref[lem:polarity-reduction]{Polarity Reduction Lemma}, we can infer the following result. 

\begin{corollary}\label{lem:lex-product-general}
    For every $r \in \set{1,2,\ldots}$, there is a threshold $\omega \in \set{1,2,\ldots}$ with the following property. 
Let  $\prec$ be a linear ordering on the product
        \begin{align*}
            \structa = \prod_{i \in I} (\set{1,\ldots,n_i},<) 
        \end{align*}
        such that $\prec$ is defined by a first-order formula of quantifier rank $r$. For every unary rank $r$ type $t$ in $\structa$, the order $\prec$ restricted to tuples of type $t$ coincides with one of the lexicographic orderings on $\structa$.
    \end{corollary}
%!TEX root = main.tex
\subsection{Powers of linear orders}
\label{sec:linear-domination}

In this section, we prove a  version of the \hyperref[lem:domination]{Domination Lemma} that considers powers of finite linear orders. The difference to the scenario treated in Section \ref{sec:direct-linear} is that here we can compare different coordinates.
\begin{lemma}[Linear Domination Lemma]
\label{lem:lin-dom}
    For all $k, r \in \set{1,2,\ldots}$, there exists a threshold $\omega \in \set{1,2,\ldots}$ such that the following holds. If $\prec$ is a linear ordering on $\structa^k$ with
    \begin{align*}
        \structa=(\set{1,\ldots,n}, <) \quad \text{and} \quad n>\omega 
    \end{align*}
    such that $\prec$ is defined by a first-order formula of rank $r$, then for every $k$-ary rank $\omega$ type $t$, there are $d \in \set{1,\ldots,k}$ and $p \in \set{-\omega,\omega}$ such that  
    \begin{align*}
         x[d] <^p y[d] \quad \text{implies} \quad   x \prec y \qquad \text{for all }  x,  y \in \structa^k \text{ of type }t.
    \end{align*}
    \end{lemma}
In the proof, we do a detailed case analysis of the expressive power of first-order logic on 
    linear orderings, which relies on the \hyperref[lem:threshold]{Threshold Lemma}.
    
For $m \in \set{1,2, \ldots}$, a tuple $x \in \structa^k$ is called \emph{$m$-separated} if for all $i\neq j$ in $\set{0,\ldots,k+1}$ it holds that $x[i]<^{m}x[j]$ or $x[i]<^{-m}x[j]$, with the convention that $x[0]= min$ and $x[k+1]= max$.
 We can extend the notion of being separated to types. A rank $r$ type $t$ is called $m$-\emph{separated} if every tuple of type $t$ is $m$-separated, and \emph{separated} if every tuple of type $t$ is $2^r$-separated.

\begin{claim}[Linear Domination Lemma - Separated Case]
\label{lem:sep-dom}
    For all $k, r \in \set{1,2,\ldots}$, there exists a threshold $\omega \in \set{1,2,\ldots}$ such that the following holds. If $\prec$ is a linear order on $\structa^k$, where
    \begin{align*}
        \structa=(\set{1,\ldots,n}, <) \qquad \text{and} \qquad n> \omega
    \end{align*}
    such that $\prec$ is defined by a first-order formula of rank $r$, then for every separated rank $\omega$ type $t$, there are $d \in \set{1,\ldots,k}$ and $p \in \set{-\omega,\omega}$ such that  
    \begin{align*}
     x[d] <^p y[d] \quad \text{implies} \quad   x \prec y \qquad \text{for all }  x,  y \in \structa^k \text{ of type }t.
    \end{align*}
    \end{claim}   

We first show that the separated version of the lemma implies the general version.

\begin{proof}[Proof of the Linear Domination Lemma]
The proof proceeds by induction on the dimension $k$. For $k=1$, the statement holds by the \hyperref[lem:threshold]{Threshold Lemma}. Let $k \geq 2$ and $r \in \set{1,\ldots}$ and assume that the Linear Domination Lemma holds for dimension $k-1$.
Let $\omega_1$ be the value obtained using Claim \ref{lem:sep-dom} for dimension $k$ and quantifier rank $r$.
Let $\omega_2$ be the value obtained using the induction hypothesis for arity $k-1$ and quantifier rank $\omega_1+2$.
Let $\omega \coloneqq \max\set{\omega_1+2,\omega_2}$.

Let $\structa \coloneqq (\set{1,\ldots,n}, <)$ for an $n > \omega$ and consider a rank $\omega$ type $t$ of arity $k$.
Let $s$ be the rank $\omega_1$ type associated with $t$. Note that all tuples of type $t$ have type $s$ but the converse does not necessarily hold.
By the \hyperref[lem:threshold]{Threshold Lemma}, the type $s$ is entirely given by quantifier-free formulas using $<^i$ for $i \in \set{1,\ldots,2^{\omega_1}}$. If $s$ is separated, by Claim~\ref{lem:sep-dom}, the result holds for tuples of type $s$, and thus in particular for tuples of type $t$.
Otherwise, if there is a $\delta\in \set{0,\ldots, 2^{\omega_1}-1}$ such that $x[0] = min$ and $x[k+1] = max$ are exactly $\delta$ apart, then all coordinates are pairwise at most $\delta$ apart and thus, $s$ (and therefore also $t$) fixes all coordinates, such that the conclusion from the lemma trivially holds. Thus, assume that there are $i\in \set{1,\ldots, k}$, $j\in \set{0,\ldots,k+1}\setminus \set{i}$ and $\delta\in \set{0,\ldots, 2^{\omega_1}-1}$ such that $x[i]$ and $x[j]$ are exactly $\delta$ apart. Without loss of generality, we assume $i=k$, $x[j]<^{\delta} x[k]$ and $x[j]\not<^{\delta + 1} x[k]$.

Let $\pi \colon \structa^k \rightarrow\structa^{k-1}$ be the projection to the first $k-1$ coordinates. Define a linear ordering $\prec_{{k-1}}$ over the tuples of $\structa^{k-1}$ which are images of tuples of type $s$ with respect to $\pi$, and such that $x\prec y\Leftrightarrow \pi(x) \prec_{k-1} \pi(y)$ for all $x$ and $y$ of type $s$. This order can be defined using a formula of quantifier rank $\omega_1+2$, simply by existentially quantifying over the missing coordinates.

Moreover, by continuity of $\pi$, all tuples of $\omega$-type $t$ are mapped to tuples of $\omega$-type $t'$.
By the induction hypothesis, we know that there are $d\in \set{1,\ldots,k-1}$ and $p\in \set{-\omega_2,+\omega_2}$ such that:
$$x[d] <^p y[d] \quad \text{implies} \quad   x \prec_{{k-1}} y \qquad \text{for all }  x,  y \in \structa^{k-1} \text{ of type }t'.$$
Thus, we have in particular:
$$x[d] <^p y[d] \quad \text{implies} \quad   x \prec y \qquad \text{for all }  x,  y \in \structa^{k} \text{ of type }t,$$
which concludes the proof.
\end{proof}

The proof of Claim \ref{lem:sep-dom} has three stages, depending on whether the arity $k$ is $1$, $2$ or bigger. The case $k=1$ is actually trivial, and the most interesting case is arity $k=2$.

%!TEX root = main.tex
\subsubsection*{Arity two}
\label{sec:dominationk2}
We first prove Claim \ref{lem:sep-dom} for $k = 2$. We need to show that there are $d \in \set{1,2}$,  $\omega>0$ and $p \in \set{-\omega,\omega}$ such that for every separated type $t$ of arity two (we use the name \emph{binary} from now on) and rank $r$, we have:
\begin{align*}
     x[d] <^p  y[d] \quad \text{implies} \quad \bar x \prec \bar y \qquad \text{for all $\bar x, \bar y \in \set{1,\ldots,n}^2$ of type $t$.}
\end{align*} 

For two pairs $ x, y$, we say that they are $\omega$-\emph{distant} if for $i\in \set{1,2}$, it holds that $x[i]<^{\pm\omega}y[i]$ or $x[i]=y[i]$.
We first show the following claim.
\begin{claim}
    If the statement from Claim \ref{lem:sep-dom} holds in the case of $\omega$-distant tuples, then it holds for all tuples.
\end{claim}
\begin{proof}
    This is shown in exactly the same way as the \hyperref[lem:polarity-reduction]{Polarity Reduction Lemma}. If we have two tuples $ x, y$ of the same type of sufficiently high rank, then we can ensure that there is a sufficiently long sequence $ x_0,\ldots, x_{\ell}$, such that the $r$-type of $(x,y)$ is the same as the one of  $( x_{i-1}, x_{i})$, for $i\in \set{1,\ldots, \ell}$. If the sequence is long enough, then $x_0, x_{\ell}$ are $\omega$-distant.
\end{proof}

Let $\omega \coloneqq 2^r$.
      Consider tuples $x$ (first row) and $y$ (second row) where the first coordinate of $ y$ is at least $2^r$ larger than the second coordinate of $ x$, as in the following picture.
        \mypicdomi{9}
        By the \hyperref[lem:threshold]{Threshold Lemma}, the order relationship $x  \prec  y$ does not depend on the choice of $x$ and $y$ (subject to the requirements in the picture above).  
        There are two cases, namely 
        \begin{align*}
            \underbrace{ x \prec  y}_{\text{A1}} \qquad    \underbrace{ x \succ  y}_{\text{A2}}.
        \end{align*}
        The cases are symmetric, we assume A1 without loss of generality.

        Consider the case as above with $x[1]<^\omega x[2]\leq y[1]<^\omega y[2]$, but the distance between $x[2]$ and $y[1]$ is $\delta<2^r$. We use a simplified version of the \hyperref[lem:polarity-reduction]{Polarity Reduction Lemma} in the case of a linear order. Let $ z\coloneqq (y[2]+\delta,y[2]+\delta+\omega)$, hence the $r$-type of $( x, y)$ is equal to the one of $( y,z)$. Then $x$ and $z$ are in situation A1, which means by the transitivity of $\prec$ that $ x\prec  y$.

        Consider the case with $ x[1]<^\omega y[1]< x[2]<^\omega y[2]$. By the \hyperref[lem:threshold]{Threshold Lemma}, we can assume that the distance $\delta$ between $x[2]$ and $ y[1]$ is at most $2^r$.
        Let $ z \coloneqq (y[2]-\delta,y[2]-\delta+\omega)$, hence the $r$-type of $( x, y)$ is equal to the one of $(y,z)$. By transitivity of $\prec$, we have $ x \prec  y \Leftrightarrow  x\prec  z$.
        Now we have that $ z[1]- x[2]=y[2]-\delta-x[2]\geq 0$ since $ x[2]<^\omega y[2]$. Using the \hyperref[lem:threshold]{Threshold Lemma}, we can assume that $ x[1]<^\omega x[2]\leq z[1]<^\omega z[2]$, which means, according to the previous paragraph, that $ x\prec  y$.

        Since we only compare distant separated tuples, the only remaining case is the one illustrated below.
        Consider now two tuples $ x$ and $ y$  that are related as follows:
        \mypicdomi{17}
        Again there are  two cases, namely 
        \begin{align*}
            \underbrace{ x \prec y}_{\text{B1}} \qquad \text{ and } \qquad  \underbrace{ x \succ  y}_{\text{B2}}.
        \end{align*}  
        Case B1 implies that the dominating coordinate is the first one and Case B2 the second.   
%!TEX root = main.tex
\subsubsection*{General arity}
In this section we complete the proof of the Separated Linear Domination  Lemma. The main idea is that it suffices to compare tuples which differ in at most two coordinates.

Let $\prec$ be a linear order on $k$-tuples in $\set{1,\ldots,n}$ that is defined by a first-order formula of quantifier rank $r$. Let $t$ be a separated $k$-ary rank $r$ type and let $\omega=2^{r}$.
 We prove that there is a dominating coordinate $d \in \set{1,\ldots,k}$ and a $p \in \set{-\omega,\omega}$ such that 
\begin{align}\label{eq:target-linear-domination}
   x[d] <^p  y[d] \quad\text{implies}\quad  x \prec  y \qquad 
    \text{for all $\bar x, \bar y \in \set{1,\ldots,n}^k$ of type $t$.}
\end{align}

Choose distinct coordinates $i,j \in \set{1,\ldots,k}$ and let  $ z$ be a tuple of  type $t$. Define 
\begin{align*}
    T^{ z}_{ij} \coloneqq \{x \in \set{1,\ldots,n}^k : \ &\text{$ x$ has type $t$ and agrees with $z$ on all coordinates except for} \\&\text{possibly $i,j$}\}.
\end{align*}
By the case of arity two, we have the following result:

\begin{claim}
For every $ z$ and distinct coordinates $i,j$ there is a dominating coordinate $d \in \set{i,j}$ and a polarity $p \in \set{-\omega,\omega}$ such that 
\begin{align*}
     x[d] <^p  y[d] \quad\text{implies}\quad  x \prec  y \qquad 
    \text{for all $ x,  y \in T^{ z}_{ij}$ of type $t$.}
\end{align*}

\end{claim} 
\begin{proof}
Without loss of generality, we assume $z[1]\leq \ldots \leq z[k]$. Let $i<j$. Since $T^{ z}_{ij} = T^{ z}_{ji}$, we can assume without loss of generality that $i < j$. We make a case analysis depending on whether $i+1=j$ holds or not. Assume $i+1=j$, and let $\structb \coloneqq \set{z[i-1]+1,\ldots,z[j+1]-1}$. We define the partial function $\pi$ from tuples of type $t$ of $\structa^k$ to pairs of elements in the universe of $\structb$ by keeping only the coordinates $i$ and $j$. We also consider the function $\sigma \colon \structb^2\rightarrow \structa^k$, which just fills in the missing coordinates by the coordinates of $z$. Note that the function $\sigma$ is continuous, hence $\prec_\structb=\sigma^{-1}(\prec)$ can be defined by a formula of quantifier rank $r$. Moreover, $\pi$ is also continuous, thus the tuples in the image of $\pi$ have the same $r$-type $t'$.
Therefore, applying the \hyperref[lem:lin-dom]{Linear Domination Lemma} to the case of dimension 2, we obtain that there are $d\in \set{1,2}$, $\omega\in \set{1,2,\ldots}$, and $p\in \set{-\omega,+\omega}$ such that:
$$x[d+1] <^p y[d+1] \quad \text{implies} \quad   x \prec_{{k-1}} y \qquad \text{for all }  x,  y \in \structb^2 \text{ of type }t'.$$
Thus, we have in particular:
$$x[d+i] <^p y[d+i] \quad \text{implies} \quad   x \prec y \qquad \text{for all }  x,  y \in T^{ z}_{ij} \text{ of type } t.$$

Similarly, if $i+1<j$, we define the structure $\structb \coloneqq \set{z[i-1]+1,\ldots,z[i+1]-1}\times \set{z[j-1]+1,\ldots,z[j+1]-1}$. Using the same arguments and Lemma \ref{lem:lex-product}, we obtain the result.
\end{proof}

It is not hard to see that there is exactly one possibility for $d$ and $p$ -- once $z$, $i$ and $j$ have been fixed -- since otherwise we would get a cycle for the order $\prec$. 
By the \hyperref[lem:threshold]{Threshold Lemma}, the dominating coordinate $d$ depends only on $i,j$ and not on the choice of $ z$, and therefore we can write $d_{ij}$ for the dominating coordinate that is appropriate to coordinates $i,j$. Also the polarity $p$ depends only on $i$ and $j$. Let us write $i \stackrel p \to j$ if, whenever the values at coordinates $i,j$ are distinct, then $j$ is the  dominating coordinate for $i,j$ and the associated polarity is $p$.

\begin{claim}\label{lem:except-two-coordinates} If $i,j,\ell$ are distinct, then $i \stackrel p \to j \stackrel q \to \ell$ implies $i \stackrel q \to \ell$. 
\end{claim}
\begin{proof}
        Choose $s \in \set{-\omega,\omega}$ arbitrarily. 
        Consider a tuple $ x$ which is $(2^{r+1})$-separated. This tuple has type $t$, and shifting any coordinate by offset in $\set{-\omega,\omega}$ still leads to a tuple that has type $t$, because the quantifier rank of $t$ is $r$. Define $ y$ to be the tuple obtained from $ x$ by adding $s$ to coordinate $i$ and adding $p$ to coordinate $i$, and define $ z$ to be the tuple obtained from $ x$ by adding  $s$ to coordinate $i$ and adding  $q$ to coordinate $\ell$. Here is a picture where $p=q=\omega$ and $s=-\omega$.

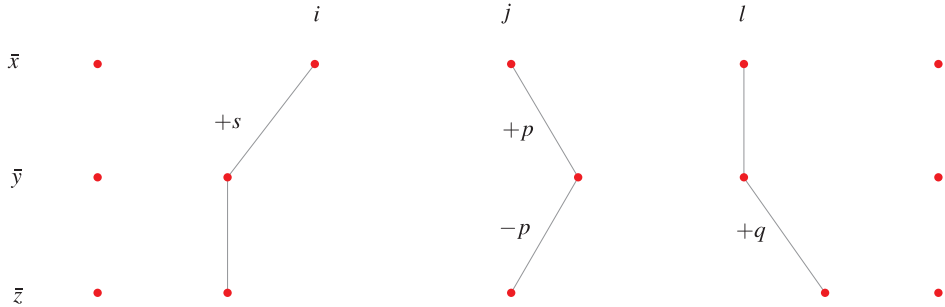
        \mypicdomi{25}

        From the assumption of the claim it follows that $ x \prec  y \prec  z$, and this holds regardless of the choice of $s$. It follows that $i \stackrel q \to \ell$. 
\end{proof}

From the above lemma it follows that the relation $i \to j$ defined by 
\begin{align*}
    i = j \qquad \text{or} \qquad i \stackrel p \to j \text{ for some }p \in \set{-\omega,\omega}
\end{align*}
is a linear order on the coordinates $\set{1,\ldots,k}$. Let $d$ be the maximal element according to this total order. Let $d'$ be the second-to-maximal element in the total order, and let $p$ be such that $d' \stackrel p \to d$ holds. From Claim~\ref{lem:except-two-coordinates} it follows that $i \stackrel p \to d$ holds for all $i \in \set{1,\ldots,k} \setminus \set{d}$. 
We show that coordinate $d$ dominates when comparing tuples where the values of coordinate $d$ are sufficiently far apart.

To finish the proof of the \hyperref[lem:lin-dom]{Linear Domination Lemma}, we prove below~\eqref{eq:target-linear-domination} for polarity $2pk$, i.e., we show
    \begin{align*}
         x[d] <^{2pk}  y[d] \quad\text{implies}\quad  x \prec  y \qquad \text{for $ x,  y \in \structa^k$  of type $t$.}
    \end{align*}
    Assume without loss of generality that the coordinates in some (equivalently, every) tuple of type $t$ are ordered so that they are strictly increasing. 
    Let $x$ and $ y$ be tuples as in the assumptions of the claim. 

    Let $x$ and $ y$ be such that $ x[d] <^{2pk}  y[d]$. We need to show $ x \prec  y$. 
    By the \hyperref[lem:threshold]{Threshold Lemma}, we can assume that all coordinates in the tuples  $ x$ and $ y$ avoid the first and last $k \cdot 2^\omega$ positions.  To prove $ x \prec  y$, we will  find  a chain of $2k$ tuples that begins in $ x$, ends in $ y$, and is growing with respect to $\prec$.  We do the proof in the case where $k=7$ and $d=4$, but the general case works the same. Define 
    \begin{align*}
        \red{z_1}= 2^\omega \quad \red{z_2} = 2\cdot 2^\omega \quad \red{z_3}= 3 \cdot 2^\omega \qquad \red{z_5} = n - 3 \cdot 2^\omega  \quad \red{z_6} = n - 2 \cdot 2^\omega  \quad \red{z_7} = n - 2^\omega.
    \end{align*}
In general,  $\red{z_i}$ is defined as $i \cdot 2^\omega$  when $i \neq d$ and otherwise it is defined as $n - (k-i+1)\cdot 2^\omega$.
The choice of coordinates $\red{z_i}$ is made so that they are far apart, and furthermore $\red{z_i}$ is to the left/right of the tuples $ x,  y$, depending on whether $i<d$ or $i>d$.   The chain that witnesses $ x \prec  z$ is given below:
    \begin{align*}
        \begin{array}{ccccccccc}
     x_1 = &        (x_1 ,& x_2 ,& x_3 ,& x_4 ,& x_5 ,& x_6 ,& x_7 ) &  \\
     x_2 = &       (\red{z_1} ,& x_2 ,& x_3 ,& x_4+p ,& x_5 ,& x_6 ,& x_7 ) &  \\         
     x_3 = & (\red{z_1} ,& \red{z_2} ,& x_3 ,& x_4+2p ,& x_5 ,& x_6 ,& x_7 ) &  \\            
     x_4 = & (\red{z_1} ,& \red{z_2} ,& \red{z_3} ,& x_4+3p ,& x_5 ,& x_6 ,& x_7 ) &  \\            
     x_5 = & (\red{z_1} ,& \red{z_2} ,& \red{z_3} ,& x_4+4p ,& \red{z_5} ,& x_6 ,& x_7 ) &  \\
     x_6 = &          ( \red{z_1} ,& \red{z_2} ,& \red{z_3} ,& x_4+5p ,& \red{z_5} ,& \red{z_6} ,& x_7 ) &  \\
     x_7 = &
     (\red{z_1} ,& \red{z_2} ,& \red{z_3} ,& x_4+6p ,& \red{z_5} ,& \red{z_6} ,& \red{z_7} ) &  \\
      y_7 = & (\red{z_1} ,& \red{z_2} ,& \red{z_3} ,& y_4-6p ,& \red{z_5} ,& \red{z_6} ,& \red{z_7} ) &  \\
      y_6 = & (\red{z_1} ,& \red{z_2} ,& y_3 ,& y_4-5p ,& \red{z_5} ,& \red{z_6} ,& \red{z_7} ) &  \\
      y_5 = & (\red{z_1} ,& y_2 ,& y_3 ,& y_4-4p ,& \red{z_5} ,& \red{z_6} ,& \red{z_7} ) &  \\
      y_4 = & (y_1 ,& y_2 ,& y_3 ,& y_4-3p ,& \red{z_5} ,& \red{z_6} ,& \red{z_7} ) &  \\ 
      y_3 = & (y_1 ,& y_2 ,& y_3 ,& y_4-2p ,& y_5 ,& \red{z_6} ,& \red{z_7} ) &  \\ 
      y_2 = & (y_1 ,& y_2 ,& y_3 ,& y_4-p ,& y_5 ,& y_6 ,& \red{z_7} ) &  \\ 
      y_1 = & (y_1 ,& y_2 ,& y_3 ,& y_4 ,& y_5 ,& y_6 ,& y_7 ) &
        \end{array}
    \end{align*}
    All tuples in the above chain have type $t$, by the assumption that coordinates $\red{z_i}$ are far apart and to the left/right of the tuples $ x,  y$. Since the dominating coordinate is incremented as the chain progresses, and at most two coordinates change in each step, we can use the assumption that coordinate $d$ dominates when only two coordinates change to conclude that each consecutive step yields a tuple that is bigger with respect to $\prec$. By transitivity, it follows that $ x =  x_1 \prec  y_1 =  y$.

%!TEX root = main.tex
\subsection{Direct products of powers of linear orders}
\label{sec:products-of-powers-of-linear-orders}
In this section we prove the most general version of the \hyperref[lem:technical-domination]{Product Domination Lemma} for linear orderings, namely the case of direct products of powers of linear orderings. 
\begin{lemma}
    \label{lem:products-of-powers-of-linear-orders}
        For all $k, r \in \set{1,2,\ldots}$, there exists a threshold $\omega \in \set{1,2,\ldots}$ such that the following holds. If $\prec$ is a linear order on 
        \begin{align*}
            \structa= \prod_{i \in I}(\set{1,\ldots,n_i}, <)^{k_i}
        \end{align*}
        such that $\prec$ is defined by a first-order formula of rank $r$, then for every unary rank $\omega$ type $t$ over $\structa$, there are $d \in I$, $e \in \set{1,\ldots,k_i}$ and  $p \in \set{-\omega,\omega}$ such that  
        \begin{align*}
            x[d][e] <^p y[d][e] \quad \text{implies} \quad   x \prec  y \qquad \text{for all $x,  y \in \structa$ of type $t$}.
        \end{align*}
        \end{lemma}

        Let $\omega$ be a threshold that is large enough -- we will specify the required bounds during the proof. Fix for the rest of this proof a unary rank $\omega$ type $t$ in $\structa$. Our goal is to find a dominating coordinate $(d,e)$ and a polarity $p$ as in the statement of Lemma~\ref{lem:products-of-powers-of-linear-orders}.

        The general strategy is as follows. 
We begin in Section~\ref{section:section-i} by looking, for every $i \in I$, at the projection
\begin{align*}
    \pi_i \colon \structa \to {\overbrace{(\set{1,\ldots,n_i}, <)^{k_i}}^{\structa_i}}  \qquad x \mapsto x[i].
\end{align*}
We  show that there exists $\sigma_i \colon \structa_i \to \structa$ which is a \emph{section} of $\pi_i$ in the sense that $\pi_i \circ \sigma_i$ is the identity on $\structa_i$. By applying the results from Section~\ref{sec:linear-domination} about powers of linear orders, we show that there is a dominating coordinate $d_i$ which works for elements in the image of the section. 
Next, in Section~\ref{sec:diagonal-section}, we consider the projection
\begin{align*}
    \pi \colon \structa \to \overbrace{\prod_{i \in I} \set{1,\ldots,n_i}}^{\structb}  \qquad x \mapsto (x[d_i])_i,
\end{align*}
which is defined in terms of the dominating coordinates $\set{d_i}_{i \in I}$ that were found in Section~\ref{section:section-i}.   Again, we find a section $\sigma \colon \structb \to \structa$. By applying the results from Section~\ref{sec:direct-linear} about direct products of linear orders, we find  a dominating coordinate $d \in I$ which works for elements in the image of the section. 
Finally, in Section~\ref{sec:proof-of-products-of-powers-of-linear-orders}, we combine the  results about the sections $\sigma_i$ and $\sigma$ to prove the conclusion of Lemma~\ref{lem:products-of-powers-of-linear-orders} for $e = d_i$ where $i=d$. 

\paragraph{Sections of $\pi_i$}
\label{section:section-i}
For each $i \in I$, apply Lemma~\ref{lem:lin-dom} about domination for powers of linear orders to $r$ and $k_i$, leading to some threshold $\omega_i$. Define 
\begin{align*}
    \omega_* \coloneqq \max (\set{r} \cup \set{\omega_i}_{i \in I}).
\end{align*}
Our first condition on the threshold $\omega$ is that
\begin{align}\label{eq:trheshold-star}
    \omega \ge \omega_*.
\end{align}
Let $t_*$ be the information of rank $\omega_*$ that is stored in the rank $\omega$ type $t$, i.e.~$t_*$ is the unique rank $\omega_*$ type contained in type $t$.  

Let $i \in I$. The general idea in this part of the proof is to study the order $\prec$ when comparing elements of $\structa$ that agree on all coordinates other than $i$. Consider the projection
\begin{align*}
    \pi_i \colon \structa \to \structa_i.
\end{align*}
For $z \in \structa$ of type $t_*$ define 
\begin{align*}
    \sigma^z_i \colon \structa_i \to \structa
\end{align*}
to be the function which fills in the missing coordinates $j \in I - \set i$ by the values used in $z$. This function is a section of $\pi_i$ in the sense that $\pi_i \circ \sigma_z$ is the identity on $\structa_i$.    Consider the preimage of $\prec$ under this section, i.e.~the relation $\prec^z_i$ defined by
\begin{align*}
    x \prec^z_i y \quad \text{if} \quad \sigma^z_i(x) \prec \sigma^z_i(y).
\end{align*}
By Theorem \ref{thm:mostowski}, $\sigma^z_i$ is continuous, and therefore  $\prec_i$ is defined by a first-order formula of same quantifier rank as $\prec$, namely $r$. Therefore, since the type $t_*$ has rank at least $\omega_i$ as obtained from Lemma~\ref{lem:lin-dom}, it follows that there are a polarity $p_i \in \set{-\omega_i,\omega_i}$ and a dominating coordinate $d_i \in \set{1,\ldots,k_i}$  such that
    \begin{align*}
        x[d_i] <^{p_i} y[d_i] \quad \text{implies} \quad x \prec_i^z y \qquad \text{for  all $x,y \in \structa_i$ of type $t_*[i]$}.
    \end{align*}    
    By Theorem \ref{thm:mostowski}, if $z$ and $z'$ have the same type of rank $r$ then $\prec^z_i$ and $\prec^{z'}_i$ are the same order. Therefore, since the quantifier rank of $t_*$ is at least $r$, it follows that the dominating coordinate $d_i$ and polarity $d_i$ do not depend on $z$, as long as it has type $t_*$.
 Because $\sigma^z_i$ is a section of  $\pi_i$ and it preserves the appropriate orderings, it follows that 
\begin{align*}
    x[i][d_i] <^{p_i} y[i][d_i] \quad \text{implies} \quad x \prec y \qquad \text{for   all  $x,y$ in the image of $\sigma^z_i$ with type $t_*[i]$}.
\end{align*}  
By applying the above to  $z=x$, we see that coordinate $d_i$ dominates whenever $x,y$ agree on coordinates other than $i$:
\begin{align}\label{eq:section-i}
    x[i][d_i] <^{p_i} y[i][d_i] \quad \text{implies} \quad x \prec y \qquad 
    \begin{array}{l}
        \text{for  all  $x,y \in \structa$ of type $t_*$}\\
        \text{such that $x[j]=y[j]$ for $j \neq i$}.
    \end{array}
\end{align}

\paragraph{Section of $\pi$}
\label{sec:diagonal-section}
As announced in the proof strategy, we now consider the projection which uses only the dominating coordinates $d_i$ that were found in Section~\ref{section:section-i}:
\begin{align*}
    \pi \colon \structa \to \overbrace{\prod_{i \in I} (\set{1,\ldots,n_i},<)}^{\structb}  \qquad x \mapsto (x[d_i])_i.
\end{align*}
We  will find a suitable section $\sigma$ of $\pi$ and prove that there is a dominating coordinate for the image of that section. 

Recall the type $t_*$ discussed in Section~\ref{section:section-i}, which is obtained by keeping only the rank $\omega_*$ information from the type $t$. Define 
\begin{align*}
    \sigma \colon \structb \to \structa
\end{align*}
to be the section of $\pi$ that is defined by
\begin{align*}
    x \in \structb \quad \mapsto \quad \begin{cases}
        \text{$\prec$-least element of $\pi^{-1}(x)$ with type $t_*$} & \text{if there is such an element}\\
        \text{$\prec$-least element of $\pi^{-1}(x)$} & \text{otherwise.}
    \end{cases}
\end{align*}

We argue why the section $\sigma$ is $(\omega_*+2k+r)$-continuous. Consider a subset $S$ of the universe of structure $\structa$ defined by a unary formula of some quantifier rank $q$. We want to show that the set of elements of $\structb$ that are sent to $S$ satisfy some formula of rank $q+\omega_*+2k+r$. This amounts to showing that, for a unary rank $q$ type $t$ of $\structa$, there is a unary formula $\phi$ over $\structb$ which selects elements whose image with respect to $\sigma$ has type $t$. Using Theorem~\ref{thm:mostowski}, the formula $\phi$ can be defined by quantifying over the missing coordinates and then checking that the whole tuple satisfies $t$, $t_*$ if possible, and that it is the minimal one to do so. Thus we obtain a formula of quantifier rank $q+\omega_*+2k+r$ (the $2$ comes from the fact that we need to check minimality).

Define $\prec_\structb$ as the inverse image of $\prec$ with respect to $\sigma$.  From continuity, it follows that  $\prec_\structb$ is defined by a first-order formula  of quantifier rank at most $\omega_* + 2k +2r$.  Apply Lemma~\ref{lem:lex-product} to this quantifier rank, yielding  some threshold $\omega_\structb$. We assume that 
 \begin{align}\label{eq:threshold-b}
     \omega \ge \omega_\structb.
 \end{align}
 Since the  projection $\pi$ is continuous, it follows that  all elements of type $t$ in $\structa$ are mapped by $\pi$ to elements of the same rank $\omega$ type, call it $t_\structb$. By Lemma~\ref{lem:lex-product} and the assumption~\eqref{eq:threshold-b}, there are a dominating coordinate $d \in I$ and a polarity $p_{\structb} \in \set{-\omega_\structb, \omega_\structb}$ such that 
 \begin{align*}
    x[d] <^{p_\structb} y[d] \quad \text{implies} \quad x \prec_\structb y \qquad \text{for  all $x,y \in \structb$ of type $t_\structb$}.
\end{align*}  
Define $e$ to be $d_i$ for $i=d$. 
Because $\sigma$ is a section and it preserves the appropriate orderings, it follows that  
\begin{align}\label{eq:section-diag}
    x[d][e] <^{p_\structb} y[d][e] \quad \text{implies} \quad x \prec y \qquad \text{for  all $x,y$ of type $t$ in the image of $\sigma$}.
\end{align}

\paragraph{Proof of Lemma~\ref{lem:products-of-powers-of-linear-orders}}
\label{sec:proof-of-products-of-powers-of-linear-orders} 
We now complete the proof of Lemma~\ref{lem:products-of-powers-of-linear-orders}.  Define $p \coloneqq 2p_i + p_\structb$. 
Let $x,y \in \structa$ have type $t$ and assume that
\begin{align}\label{eq:assumption-two-p}
        x[d][e] <^{p} y[d][e].
\end{align}
To prove $x \prec y$, as required in the conclusion of Lemma~\ref{lem:products-of-powers-of-linear-orders}, we will find an $\prec$-ascending chain  which begins in $x$, ends in $y$, and such that  each step in the ascending chain is proved using the results from Sections~\ref{section:section-i} and~\ref{sec:diagonal-section}.  

\begin{claim}\label{cl:canonical}
    There is an $x'$ of type $t$ in the image of $\sigma$ such that $x \prec x'$ and 
    \begin{align*}
        x'[i][d_i] = x[i][d_i] + p_i \qquad \text{for all $i \in I$.}
    \end{align*}
\end{claim}
\begin{proof}
    We say that $x$ is \emph{canonical on coordinate $i \in I$} if $x[i]=z[i]$ for some $z$ in the image of $\sigma$. By Theorem \ref{thm:mostowski}, if $x$ is canonical on all coordinates $i \in I$, then it is in the image of $\sigma$.   Therefore, we can prove the claim by induction on the number of coordinates $i \in I$ on which $x$ is canonical, and in the induction step we can make a single coordinate $i$ canonical, at the cost of shifting $x[i][d_i]$ by $p_i$ positions, thanks to~\eqref{eq:section-i}.
\end{proof}

Apply Claim \ref{cl:canonical} to $x$, yielding some $x'$ of type $t$ with $x \prec x'$. Apply a symmetric result to $y$, yielding some $y'$ of type $t$ with  $y' \prec y$ and 
\begin{align*}
    y'[i][d_i] = y[i][d_i] - p_i \qquad \text{for all $i \in I$}.
\end{align*}
By definition of $p$ and the assumption~\eqref{eq:assumption-two-p}, we see that 
\begin{align*}
    x'[i][d_i] <^{p_\structb} y'[i][d_i] \qquad \text{for $i=d$}
\end{align*}
and therefore~\eqref{eq:section-diag} can be applied to  conclude $x' \prec y'$, and thus also $x \prec y$. 
%!TEX root = main.tex
\subsection{Proof of the Product Domination Lemma}
\label{sec:pre-final-proof-of-domination}
In this section, we complete the proof of the \hyperref[lem:technical-domination]{Product Domination Lemma}, and therefore also of the \hyperref[lem:domination]{Domination Lemma}. 

Let $\omega$ be a threshold that is high enough, we will specify the lower bounds on $\omega$ throughout the proof. Let $t$ be a unary rank $\omega$ type in $\structa$.
Consider the projection 
\begin{align*}
    \pi \colon \structa \to \overbrace{\prod_{i \in I} \set{1,\ldots,n_i},<)^{k_i}}^{\structb}
\end{align*}
which maps each element of $\structa$ to the appropriate tuple of block numbers. This function is continuous.
Let $\omega_*$ be the threshold obtained by applying Lemma~\ref{lem:products-of-powers-of-linear-orders} to quantifier rank $r$ and the product $\structb$. We assume that 
\begin{align}
    \omega \ge \omega_*.
\end{align}
Let $t_*$ be the type of rank $\omega_*$ which stores the quantifier rank $\omega_*$ information of type $t$.

We say that $x,x' \in \structa$ \emph{overlap} if there is a block which intersects both $x$ and $x'$. More formally,
\begin{align*}
    (\pi(x))[i][j] = (\pi(x'))[i][j'] \quad \text{for some $i \in I$ and $j,j' \in \set{1,\ldots,k_i}$.}
\end{align*}
Note that overlapping is defined purely in terms of the image under $\pi$. The first step in the proof is the following claim, which shows that there is a dominating coordinate when only comparing non-overlapping elements. 

\begin{claim}\label{claim:overlap}
    There  exist $d \in I$, $e \in \set{1,\ldots,k_d}$ and $q \in \set{-\omega_*,\omega_*}$ such that
    \begin{align*}
        x[d][e] \sqsubset^{p_{\structb}} y[d][e] \quad \text{implies} \quad x \prec y \qquad \text{for  all non-overlapping $x,y \in \structa$ of type $t_*$}.
    \end{align*}  
\end{claim}
\begin{proof} One can show that there is a continuous section $\sigma$ of $\pi$ such that $\sigma \circ \pi$ preserves the type $t_*$, i.e.,~it maps  $t_*$ to a subset of itself. To define the section, one only needs to choose for each coordinate $(i,j)$ and each block of $\structa_i$ an appropriate representative such that tuples of type $t_*$ are mapped to tuples of type $t$. Using compositionality, we thus have that $\sigma$ maps tuples of the same type to tuples of the same type.

     Define  $\prec_\structb$ to be the pre-image of $\prec$ under this section. Since $\sigma$ is continuous, the order $\prec_\structb$ is  definable using quantifier rank $r$. By Lemma~\ref{lem:products-of-powers-of-linear-orders} and the definition of $\omega_*$,  there exist dominating coordinates $d \in I$, $e \in \set{1,\ldots,k_d}$ and a polarity $q \in \set{-\omega_*,\omega_*}$ such that
\begin{align}\label{eq:section-general}
    x[d][e] \sqsubset^q y[d][e] \quad \text{implies} \quad x \prec y \qquad \text{for  all $x,y \in \structa$ of type $t_*$ in the image of $\sigma$}.
\end{align}  
We now extend the above result to non-overlapping elements of type $t_*$, as required in the statement of the claim.
Using compositionality, one can show that if $x$ and $y$ are non-overlapping, then the binary rank $r$ type of the pair $(x,y)$ in $\structa$ is uniquely determined by the unary rank $r$ types of $x$ and $y$ as well as the binary rank $r$ type of $\pi(x,y)$. 
Since $\sigma \circ \pi$ does not change the value under $\pi$, it follows that $x$ and $y$ overlap if and only if their images with respect to $\sigma \circ \pi$ overlap. As we have argued at the beginning of this proof, $\sigma \circ \pi$ maps the set of tuples of type $t_*$ to a subset of itself, and therefore if $x$ and $y$ have type $t_*$, then also their images under $\sigma \circ \pi$ have type $t_*$. It follows that 
\begin{align*}
    (x,y) \equiv_{\omega_*} (\sigma \circ \pi)(x,y) \qquad \text{for all non-overlapping $x,y \in \structa$ of type $t_*$}
\end{align*}
By combining this observation with~\eqref{eq:section-general}, we obtain the conclusion of the claim.
\end{proof}

The general idea in the rest of the proof is to show that if $x,y \in \structa$ of type $t$ are possibly overlapping, then one can find a $z$ which overlaps neither with $x$ nor with $y$ and where $x \prec z \prec y$ can be shown using Claim~\ref{claim:overlap}. To find this $z$, we will shift $x$ (or $y$) by several blocks to the left or right, as explained below.

For $b \in \structb$ and  a possibly negative integer $\delta$, define $b+\delta$ to be the result of  adding $\delta$ to all coordinates of $b$. Note that $x+\delta$ might fall out of $\structb$, e.g.~because some coordinate might become negative. 
For $x \in \structa$, define $\Delta_x$ to be the set of integers $\delta$ such that
\begin{align*}
    \pi(y) = \pi(x) + \delta \qquad \text{for some $y \in \structa$ of type $t_*$}.
\end{align*}
The key observation is the following claim, which says that either $\Delta_x$ is big for all $x \in \structa$ of type $t$, or one can trivially find a dominating coordinate, because there is a choice of coordinates the values at which always lie in the same block.
\begin{claim}\label{claim:dichotomy}
    One of the following holds:
\begin{enumerate}
    \item There are $i \in I$ and $j \in \set{1,\ldots,k_i}$ such that
    \begin{align*}
        x[i][j] \not \sqsubset y[i][j] \qquad \text{for all $x,y \in \structa$ of type $t$}.
    \end{align*}
    \item For every $x \in \structa$ of type $t$, the set $\Delta_x$ contains $\set{-p,\ldots,p}$.
\end{enumerate} 
\end{claim}
\begin{proof}
    For $s \in \set{0,1,\ldots}$ define $t^s$ to be rank $s$ information stored in type $t$, i.e.~this is the unique rank $s$ type contained in $t$. By continuity of $\pi$, the image of $t^s$ under $\pi$ is a rank $s$ type   in the structure $\structb$, call it $t^s_\structb$. By the \hyperref[lem:threshold]{Threshold Lemma}, every rank $s$ type in $\structb$ amounts to measuring distances between coordinates and minimal and maximal elements, up to threshold $2^s$. We say that $t^s_\structb$ is \emph{anchored} if it fixes the distance of some coordinate to either the minimal or maximal element, i.e.~some (equivalently, every) $x$ in $\structb$ of type $t^s_\structb$ is such that $x[i][j]$ has distance $<2^s$ from either $1$ or $n_i$ for some $i \in I$ and $j \in \set{1,\ldots,n_i}$. If $t^\omega_\structb$ is anchored, then item 1 in the claim holds. Suppose that $t^{s+1}_\structb$ is not anchored. It follows from the \hyperref[lem:threshold]{Threshold Lemma} that for every $b$ of type $t^{s+1}_\structb$ and every $\delta \in \set{-2^s,\ldots,2^s}$, 
      the shifted value $b+\delta$ has type $t^{s}_\structb$. In particular, if $s \ge \omega_*$ and $t^{s+1}_\structb$ is not anchored, then $\Delta_x$ has size at least $2^{s+1}$ for every $x$ of type $t^{s+1}_\structb$. Thus, with the assumption that
      \begin{align*}
          \omega > \omega_* + \log p,
      \end{align*}
      the result follows.
\end{proof}

Apply Claim \ref{claim:dichotomy}. If the first case holds, then $i$ and $j$ are dominating coordinates by vacuous truth. We are left with the other case. We will show that $d$ and $e$, as in Claim~\ref{claim:overlap}, are dominating coordinates. Assume that $x,y \in \structa$ have type $t$ and assume $x \sqsubset^p y$. We will show $x \prec y$. By the assumption on $\Delta_x$, we know that for every integer $\delta$ with $0 < \delta <p$, there is a $x_\delta$ of type $t_*$ such that
    \begin{align*}
        \pi(x_\delta) = \pi(x) + \delta.
    \end{align*}
    For every $i \in I$ and $j \in \set{1,\ldots,k_i}$, there are at most two choices of $\delta$ such that
    \begin{align*}
        (\pi(x))[i][j] = (\pi(x_\delta))[i][j] \quad \text{or} \quad (\pi(y))[i][j] = (\pi(x_\delta))[i][j].
    \end{align*}
    By a counting argument and thanks to the definition of $p$, it follows that there is some $p_\structb$ with 
    \begin{align*}
        p_\structb < \delta < p-p_\structb
    \end{align*}
    such that $x_\delta$ overlaps neither with $x$ nor with $y$. Therefore, we can use Claim~\ref{claim:overlap} to get $x \prec x_\delta \prec y$.

\end{document}